\documentclass[11pt,a4paper]{article}
\usepackage[utf8]{inputenc}
\usepackage[english]{babel}

\usepackage[margin=1.25in]{geometry}
\usepackage{graphicx}
\usepackage{amsmath,amsthm}
\usepackage{amsfonts,amssymb,mathtools} %,bbm
\usepackage{enumitem}
\usepackage{complexity}
\usepackage{hyperref}
\usepackage{subfig}
\usepackage{authblk}
\usepackage{tikz}
\usetikzlibrary{arrows.meta,shapes}
\usetikzlibrary{shapes.geometric, circuits.logic.US, positioning}
\usetikzlibrary{decorations,decorations.pathreplacing}

\usepackage[color=red!60,textsize=small]{todonotes}

\theoremstyle{plain}
\newtheorem{theorem}{Theorem}[section]

\newtheorem{corollary}  [theorem]{Corollary}
\newtheorem{definition} [theorem]{Definition}

\newtheorem{lemma}      [theorem]{Lemma}
\newtheorem{problem}    [theorem]{Problem}
\newtheorem{proposition}[theorem]{Proposition}

\newtheorem{claim}      [theorem]{Claim}
\newtheorem{remark*}{Remark}

\newenvironment{claimproof}[1]{\par\noindent{\textit{Proof of claim.}}\space#1}{\hfill $\diamond$\vspace{0.3cm}}

\newcommand{\N} {\mathbb{N}}

\newcommand{\BCTP} {{\bf{MBAN-DCT}}}
\newcommand{\ICVP}{{\bf{iter\textnormal{-}CVP}}}
\newcommand{\IMCVP}{{\bf{iter\textnormal{-}MCVP}}}
\newcommand{\SSmC}{{\bf{MBAN\textnormal{-}LC}}}

\newcommand{\OCCM}{{\bf{DCTP-One}}}

\newcommand{\CGlobalFP}{{\bf C\textnormal{-}GlobalFP}}
\newcommand{\MCtwoFP}{{\bf MC\textnormal{-}$2$FP}}
\newcommand{\MCLC}{{\bf MC\textnormal{-}LC}}
\newcommand{\MCExpLC}{{\bf MC\textnormal{-}ExpLC}}
\newcommand{\Czerotoone}{{\bf C\textnormal{-}$0^n$to$1^n$}}
\newcommand{\MCkCC}{{\bf MC\textnormal{-}$k$CC}}
\newcommand{\MCgekCC}{{\bf MC\textnormal{-}$\geq k$CC}}

\newcommand{\PSC}{\PSPACE\textnormal{-complete}}
\newcommand{\PSH}{\PSPACE\textnormal{-hard}}
\newcommand{\NPH}{\NP\textnormal{-hard}}
\renewcommand{\NPC}{\NP\textnormal{-complete}}

\newcommand{\majloc}{\texttt{maj}}
\newcommand{\majglob}[1]{A}
\newcommand{\majority}[1]{\majloc(#1) }

\newcommand{\val}       [2] {#1_{#2} }

\renewenvironment{problem}[4]{
    \medskip
    \noindent
    \fbox{\parbox{.98\textwidth}{
        \vspace*{.1em}
        \textit{#1} (#2).\\[.2em]
        \textit{Input:} #3.\\[.2em]
        \textit{Question:} #4 ?\\[-.9em]
    }}
    \medskip
    }
	
\definecolor{myGreen}{RGB}{0,120,0}
\definecolor{myRed}{RGB}{200,0,0}

\title{Majority Boolean networks classifying density:\\structural characterization and complexity}

\author[1,2]{K{\'e}vin Perrot}
\author[1]{Marius Rolland}
\affil[1]{Aix Marseille Université, CNRS, LIS, Marseille, France}
\affil[2]{Université publique, France}
\date{2026}

\begin{document}

\maketitle

\begin{abstract}
  Given a set of entities each holding a Boolean state,
  the Density Classification Task (DCT) asks them to converge
  to the most represented state.
  Given a directed graph of entities
  where each node synchronously updates to the local majority among its in-neighbors,
  we characterize in terms of three forbidden patterns when it solves DCT,
  and show that discovering these patterns is complete for {\NP} and {\PSPACE}.
  %
  %\keywords{
  %  Majority rule \and
  %  Density classification task \and
  %  Discrete dynamical systems \and
  %  Boolean automata networks \and
  %  Computational complexity.
  %}
\end{abstract}

%%%%%%%%%%%%%%%%%%%%%%%%%%%%%%%%
% CONTENT
%%%%%%%%%%%%%%%%%%%%%%%%%%%%%%%%
	
%%%%%%%%%%%%%%%%%%%%%%%%%%%%%%%%
\section{Introduction}

A finite discrete dynamical system of $n$ interacting entities solves the
Density Classification Tasks (DCT, also called the global majority problem)
when, from any starting configuration $c$ of Boolean states,
it evolves to the homogeneous fixed point configuration consisting in the majority state in $c$~\cite{o14}.
Depending on the precise system at stake, solutions may exist or not.
The DCT is traditionally studied on Cellular Autamata (CAs),
a highly parallel model of computation.
%with surprisingly ``simple'' rules capable of universal computation~\cite{c04,r11,t22}.
Along with the parity (convergence depending on the parity of state $1$)
and synchronization (convergence to a limit cycle of two homogenous configurations) tasks,
it challenges the abilty of distributed information processing systems to reach a non-trivial consensus~\cite{mmog18,prb25,wnbb25}.

The DCT is well defined only when $n$ is odd, to avoid tie cases.
It is known since 1995 and the work of Land and Belew that
the DCT does not admit any solution in CAs with periodic boundary conditions,
regardless of the radius and dimension~\cite{lb95}.
Nevertheless, various relaxations led to solutions, consequently highlighting the locks behind this impossibility.
For example, in one dimension using few rules in temporal or spatial sequence~\cite{mo05,f97},
or with additional tracks or alphabets~\cite{kg12,bmeor13,pbr25},
or with the adjunction of randomness~\cite{f11}.
%See~\cite{o14} for a survey.

% unused references
% -----------------
% - 184 somehow solves density... (1996)
%   https://doi.org/10.1103/PhysRevLett.77.4969
% - Evolutionary algorithm to find not so bad solutions:
%   - (1995) https://melaniemitchell.me/PapersContent/EGSCA.pdf
%   - (2000) https://doi.org/10.7551/mitpress/1432.003.0060
%   - (2007) https://doi.org/10.1016/j.tcs.2007.01.001
%   - (2009) https://doi.org/10.1016/j.entcs.2009.09.018
% - Parity solution FBO (2013, corrected AUTOMATA'2025)
%   https://doi.org/10.1007/s11047-013-9374-9
% - Memory ECA 184 (2019)
%   https://wpmedia.wolfram.com/sites/13/2019/03/18-3-4.pdf
% - Noise experiment on 13 rules (2024)
%   https://doi.org/10.1162/isal_a_00823

Our work is inspired by the recent discovery of Majority Boolean Automata Networks (MBANs) solving the DCT~\cite{ao20}.
An MBAN is defined as a directed graph where each node holds a state, and updates it to the majority state among its in-neighbors
(all in-degrees should be odd, to avoid tie cases).
Of course, the clique (with loops) on $n$ vertices ($n$ being odd) is a trivial example
where the dynamics converges to the majority state in one step.
We dig into the necessary and sufficient conditions characterizing on which graphs does
the local majority rule solves the global majority problem.

After defining MBANs and the DCT (Section~\ref{sec:def}),
we first show that it is {\PSC} to decide whether a single configuration converges to the correct uniform fixed point (Section~\ref{sec:one_conf}).
Then we characterize the graphs of MBANs solving the DCT by three forbidden patterns (Section~\ref{sec:chara_MBAN_sol_DCTP}),
and study the complexity of recognizing them (Sections~\ref{sec:ss_and_lead_NPC}). % \ref{sec:part_ssmc}
%Some proofs are presented in Appendix~\ref{a:proofs}.

%%%%%%%%%%%%%%%%%%%%%%%%%%%%%%%%
\section{Definitions}\label{sec:def}

Let $[n]=\{0,\ldots,n-1\}$ be a set of entities called \emph{automata}.
A \emph{configuration} $x\in\{0,1\}^n$ assigns a Boolean state $\val{x}{v}$ to each $v\in [n]$.
We extend this notation to any subset $W\subseteq V$,
with $\val{x}{W}\in\{0,1\}^{|W|}$ the restriction of $x$ on domain $W$,
and also denote $0^n$ and $1^n$ the uniform configurations of size $n$.
Let $\majloc:\{0,1\}^n\to\{0,1\}$ denote the majority functions, that is:
\[
  \majloc(x)=\begin{cases}
    0 & \text{if } \sum_{v\in [n]} x_v < \frac{n}{2},\\
    1 & \text{if } \sum_{v\in [n]} x_v > \frac{n}{2},\\
    \text{undefined} & \text{otherwise (tie case)}.
  \end{cases}
\]
A \emph{Majority Boolean Automata Network} (MBAN) lets configurations evolve
according to a global map $A_G:\{0,1\}^n\to\{0,1\}^n$
determined by a directed graph $G=(V,E)$ with $V=[n]$, as follows:
\[
  \text{for all } v\in V,~\val{A_G(x)}{v}=\majloc(x_{N_G(v)}).
\]
In words, at each step each automaton takes the most represented state among its in-neighbors
($N_G(v)$ is the in-neighborhood of $v$ in $G$).
When $G$ is clear from the context, we may drop the subscript.
The dynamics (graph of $A_G$) of a finite discrete system has connected components made of
one \emph{limit cycle} (whose \emph{period} is the smallest $t$ such that $A^t(x)=x$ for each $x$ in the cycle)
and upward trees rooted in it (transient configurations);
It is called a \emph{fixed point} if it has period $1$.

We say that an MBAN $A$ solves the \emph{Density Classification Task} (DCT) when,
from any configuration $x$,
all automata converge to the majority state $\majority{x}$:
\[
  \forall x\in\{0,1\}^n,~\exists t\in\N,~A^t(x)=\majority{x}^n
\]
Remark that for any MBAN the configurations $0^n$ and $1^n$ are fixed points.
Moreover, the dynamics is deterministic therefore in any solution to the DCT
the majority state never changes throughout the evolution.
The DCT is well defined only when $n$ is odd, which will always be assumed.
The dynamics of the MBAN is well defined only when $N_G(v)$ is odd for all $v\in V$,
which will be ensured in all our constructions.
We focus on the following problem.

\begin{problem}
{MBAN solves DCT problem}{\BCTP}
{An MBAN $A_G$ of size $n$ given by its directed graph $G$}
{Does $A_G$ solve the DCT}
\end{problem}

The problem \BCTP{} is in \PSPACE{}, because it is answered by sequentially checking the $2^n$ configurations,
each requiring polynomial space to be iterated for $2^n$ time steps
(first recording its majority state to check the correctness of convergence).
An immediate necessary condition for $A_G$ to solve the DCT is that $G$ is strongly connected,
and all our constructions will be so.

%%%%%%%%%%%%%%%%%%%%%%%%%%%%%%%%
\section{For one configuration}\label{sec:one_conf}

Although the global property of solving DCT
is difficult to analyze, deciding whether a single given configuration
converges to the correct uniform configuration is easier to classify in terms of complexity.
We will prove that the problem \OCCM{} is \PSC{} (Theorem~\ref{th:config_converge_PSC}). 

\begin{problem}
{MBAN solves DCT on One configuration}{\OCCM}
{An MBAN $A_G$ of size $n$ given by its graph $G$, and $x\in\{0,1\}^n$}
{Does there exist an integer $t>0$ such that $A_G^t(x)=\majority{x}^n$}
\end{problem}

The membership in \PSPACE{} is given by the proof that \BCTP{} is in \PSPACE{}. 
Thus, it remains to prove that this problem is \PSH{}.
For this, we use the simulation of $\mathrm{AND}$ and $\mathrm{OR}$ gates with the majority local rule (see Figure~\ref{fig:gates}). 
Consequently, we can intuitively model a monotone circuit with an MBAN,
and we will indeed reduce from \emph{Iterated Monotone Circuit Value Problem} (\IMCVP{}) which is known to be \PSC{}~\cite{gmst16}. 
This problem is: given a monotone circuit $C : \{0,1\}^n \to \{0,1\}^n$ (without negation), a configuration $x\in\{0,1\}^n$ and an integer $i$ among $\{0,\dots,n-1\}$, does there exist an integer $t>0$ such that $C^{t}(x)_i = 1$?
That is, does iterating circuit $C$ from $x$ eventually brings the $i$-th component in state $1$?

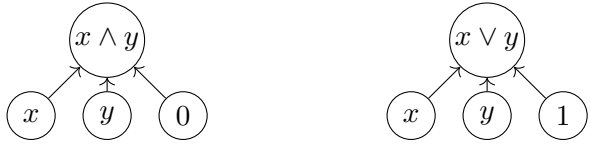
\begin{figure}[t]
    \centering
    \begin{tikzpicture}
	% styles
	\tikzstyle{automate} = [draw,circle,inner sep=1pt,minimum size=18pt]
	\tikzstyle{arc} = [-{>[length=1mm]}]

    \begin{scope}[xshift=-3cm]
        \node[automate] (ax) at (-2,-0.5) {$x$};
	    \node[automate] (ay) at (-1,-0.5) {$y$};
	    \node[automate] (a0) at (0,-0.5) {$0$};
	    \node[automate] (ar) at (-1,0.5) {$x\wedge y$};
        \draw[arc] (ax) to (ar);
        \draw[arc] (ay) to (ar);
        \draw[arc] (a0) to (ar);
    \end{scope}

    \begin{scope}
        \node[automate] (ox) at (0,-0.5) {$x$};
    	\node[automate] (oy) at (1,-0.5) {$y$};
        \node[automate] (o1) at (2,-0.5) {$1$};
    	\node[automate] (or) at (1,0.5) {$x \vee y$};
        
        \draw[arc] (ox) to (or);
        \draw[arc] (oy) to (or);
        \draw[arc] (o1) to (or);
    \end{scope}

\end{tikzpicture}
    \caption{
      Simulation of $\mathrm{AND}$ (left) and $\mathrm{OR}$ (right) gates of two Boolean inputs $x$ and $y$,
      where the top node applies a majority local rule equivalent to the simulated gate
      ($0$ and $1$ are constant Boolean values).
      %\vspace*{-1em} % HACK
    }
    \label{fig:gates}
    \label{fig:enter-label}
\end{figure}

\begin{theorem}\label{th:config_converge_PSC}
    \OCCM{} is \PSC{}. %, even when restricted to strongly connected digraphs.
\end{theorem}
    
\begin{proof}
    Let $C : \{0,1\}^n \to \{0,1\}^n$ be a monotone Boolean circuit, $x$ a configuration of $C$ and $i \in \{0,\ldots,n-1\}$.
    To construct a directed graph $G$ for the MBAN, we start from the circuit $C$ seen as an acyclic graph $(V_C,E_C)$ where each vertex is either an input (there are $n$ inputs), or a gate AND or OR (the $n$ outputs are specific gates).
    Without loss of generality, we assume that $|V_C|$ is even, that the circuit has fan-in two,
    is layered (the two in-neighbors of a gate in layer $\ell+1$ are taken from layer $\ell$)
    and that after some iterations, the formula describing the value of output $i$ depends of all the inputs.
    First, we convert all the AND and OR gates into majority local rules as described in Figure~\ref{fig:gates},
    using two new vertices $b_0$ and $b_1$ to play the role of the constants $0$ and $1$ respectively. 
    Second, for each input vertex $j\in\{0,\dots,n-1\}$ we add as in-neighbors the output vertex number $j$,
    and also $b_0$ and $b_1$ (having one extra $0$ and one extra $1$ will let the majority
    reproduce the identity, hence the ouputs will be copied to the inputs).
    We add the following in-neighbors to $b_0$ and $b_1$:
    \begin{itemize}[nosep]
        \item $b_0$ has a loop, and also $i$ and $b_1$ as in-neighbors,
        \item $b_1$ has a loop, and also $b_0$ as in-neighbor.
    \end{itemize}
    At this point all vertices have in-degree three, except $b_1$ which has in-degree two, and the graph is strongly connected.
    
    Finally, we add a clique on vertex set $P$ of size $|V_C|+1$ (odd),
    and take one arbitrary vertex $p\in P$, remove its loop, replace it with $b_1$ as in-neighbor, and also add $p$ as a third in-neighbor of $b_1$.
    This concludes the construction of $G'=(V',E')$ on vertex set $V'=V_C\sqcup\{b_0,b_1\}\sqcup P$,
    where $|V'|=2\,|V_C|+3$ is odd. All the vertices have in-degree three, except the vertices from $P$ which have in-degree $|V_C|+1$, and $G'$ is strongly connected. See Figure~\ref{fig:DCTP-One}.
    
    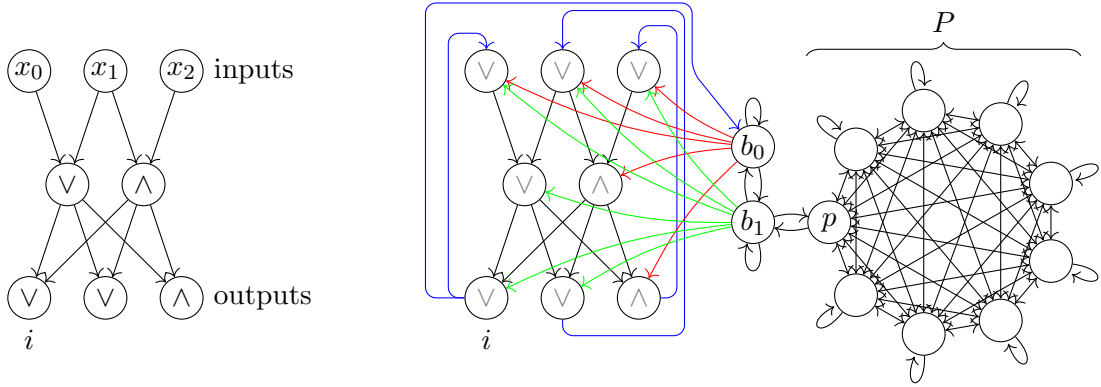
\begin{figure}[t]
        \centering
        \resizebox{\textwidth}{!}{\begin{tikzpicture}
    \tikzstyle{node} = [circle, draw, inner sep=1pt, minimum size=16pt]
    \tikzstyle{mnode} = [circle, draw, inner sep=1pt, minimum size=16pt, text=black!50]
    \tikzstyle{arc} = [-{>[length=1mm]}]
    \tikzstyle{rarc} = [-{>[length=1mm]}, rounded corners=4pt]
    \tikzstyle{darc} = [{<[length=1mm]}-{>[length=1mm]}]
    % circuit
    \begin{scope}
      % input nodes
      \foreach \x in {0,1,2}
        \node[node] (i\x) at (\x,0) {$x_{\x}$};
      % gates
      \node[node] (g0) at (0.5,-1.5) {$\vee$};
      \node[node] (g1) at (1.5,-1.5) {$\wedge$};
      % output gates
      \node[node] (o0) at (0,-3) {$\vee$};
      \node[node] (o1) at (1,-3) {$\vee$};
      \node[node] (o2) at (2,-3) {$\wedge$};
      % arcs
      \draw[arc] (i0) to (g0);
      \draw[arc] (i1) to (g0);
      \draw[arc] (i1) to (g1);
      \draw[arc] (i2) to (g1);
      \draw[arc] (g0) to (o0);
      \draw[arc] (g0) to (o1);
      \draw[arc] (g0) to (o2);
      \draw[arc] (g1) to (o0);
      \draw[arc] (g1) to (o1);
      \draw[arc] (g1) to (o2);
      % labels
      \node[right=0pt of i2] {inputs};
      \node[right=0pt of o2] {outputs};
      \node[below=0pt of o0] {$i$};
    \end{scope}
    % MBAN
    \begin{scope}[shift={(6,0)}]
      % nodes of the circuit
      \foreach \x in {0,1,2}
        \node[mnode] (mi\x) at (\x,0) {$\vee$};
      \node[mnode] (mg0) at (0.5,-1.5) {$\vee$};
      \node[mnode] (mg1) at (1.5,-1.5) {$\wedge$};
      \node[mnode] (mo0) at (0,-3) {$\vee$};
      \node[mnode] (mo1) at (1,-3) {$\vee$};
      \node[mnode] (mo2) at (2,-3) {$\wedge$};
      % label i
      \node[below=0pt of mo0] {$i$};
      % arcs of the circuit
      \draw[arc] (mi0) to (mg0);
      \draw[arc] (mi1) to (mg0);
      \draw[arc] (mi1) to (mg1);
      \draw[arc] (mi2) to (mg1);
      \draw[arc] (mg0) to (mo0);
      \draw[arc] (mg0) to (mo1);
      \draw[arc] (mg0) to (mo2);
      \draw[arc] (mg1) to (mo0);
      \draw[arc] (mg1) to (mo1);
      \draw[arc] (mg1) to (mo2);
      % i and outputs to inputs
      \draw[rarc,blue] (mo0) -- ++(-.5,0) -- ++(0,3.5) -- ++(.5,0) -- (mi0);
      %\draw[rarc,blue] (mo0) -- ++(-.6,0) -- ++(0,3.6) -- ++(1.6,0) -- (mi1);
      %\draw[rarc,blue] (mo0) -- ++(-.7,0) -- ++(0,3.7) -- ++(2.7,0) -- (mi2);
      \draw[rarc,blue] (mo1) -- ++(0,-.5) -- ++(1.6,0) -- ++(0,4.3) -- ++(-1.6,0) -- (mi1);
      \draw[rarc,blue] (mo2) -- ++(.5,0) -- ++(0,3.6) -- ++(-.5,0) -- (mi2);
      % nodes b_0 b_1
      \node[node] (b0) at (3.5,-1) {$b_0$};
      \node[node] (b1) at (3.5,-2) {$b_1$};
      % arcs b_0 b_1 to circuit
      \draw[arc,red] (b0) to[bend left=10] (mi0);
      \draw[arc,red] (b0) to[bend left=10] (mi1);
      \draw[arc,red] (b0) to[bend left=10] (mi2);
      \draw[arc,red] (b0) to[bend right=10] (mg1);
      \draw[arc,red] (b0) to[bend right=10] (mo2);
      \draw[arc,green] (b1) to[bend left=10] (mi0);
      \draw[arc,green] (b1) to[bend left=10] (mi1);
      \draw[arc,green] (b1) to[bend left=10] (mi2);
      \draw[arc,green] (b1) to[bend left=10] (mg0);
      \draw[arc,green] (b1) to[bend right=10] (mo0);
      \draw[arc,green] (b1) to[bend right=10] (mo1);
      % arcs b_0 b_1
      \draw[arc] (b0) to[bend left=15] (b1);
      \draw[arc] (b1) to[bend left=15] (b0);
      \draw[arc,loop above] (b0) to (b0);
      \draw[arc,loop below] (b1) to (b1);
      % i to b_0
      \draw[rarc,blue] (mo0) -- ++(-.8,0) -- ++(0,3.9) -- ++(3.5,0) -- ++(0,-.9) -- (b0);
      % clique (p is p4)
      \foreach \r in {0,1,2,3,5,6,7,8}
        \node[node,shift={(6,-2)}] (p\r) at (\r*40+20:1.5) {} edge [in=40*\r+20-12,out=40*\r+20+12,loop] ();;
      \node[node,shift={(6,-2)}] (p4) at (4*40+20:1.5) {$p$};
      \foreach[count=\i from 0] \r in {1,...,8}{
        \foreach \rr in {0,...,\i}{
          \draw[darc] (p\r) to (p\rr);
        }
      }
      \draw[decorate,decoration={brace,amplitude=5pt}] (4.2,.2) -- node[above,yshift=5pt]{$P$} (7.8,.2);
      % arcs p b_1
      \draw[arc] (p4) to[bend left=15] (b1);
      \draw[arc] (b1) to[bend left=15] (p4);
    \end{scope}
\end{tikzpicture}}
        \caption{
            Illustration of the construction of the digraph $G'$ of an MBAN (right), from a monotone circuit $C$ with $n=3$ and $|V_C|=8$ (left) in the proof of Theorem~\ref{th:config_converge_PSC}.
            The $\wedge$ and $\vee$ labels in the nodes of $G'$ indicate the expected gate simulated when the state of $b_0$ is $0$ and the state of $b_1$ is $1$.
            %\vspace*{-1em} % HACK
        }
        \label{fig:DCTP-One}
    \end{figure}
    
    The configuration $x'$ of $A_{G'}$ is constructed by:
    \begin{itemize}[nosep]
        \item assigning $x$ to the former input vertices of $C$,
        \item assigning $0$ to all other vertices of $V_C$,
        %then evaluating $C$ (in linear time) to assign Boolean values to all other vertices in $V_C$ according to the values computed by their respective gate in $C$ (on input $x$),
        \item assigning state $0$ to $b_0$, and state $1$ to $b_1$ and all the vertices in $P$. 
    \end{itemize}
    Configuration $x'$ has at least $|V_C|+2$ vertices in state $1$ from the last item,
    therefore $\majority{x'}=1$.
    %
    %The proof that $x'$ converges to $1^{|V'|}$ if and only if $C^t(x)_i=1$ for some $t$,
    %a state $1$ appears at output number $i$ in some iteration of $C$ starting from $x$,
    %is in Appendix~\ref{a:proofs}.

    \medskip
    We now prove that $x'$ converges to $1^{|V'|}$ if and only if a state $1$ appears at output number $i$ in some iteration of $C$ starting from $x$.
    First, observe that $b_1$ and all the vertices in $P$ will remain in state $1$ forever.
    Second, remark that as long as output number $i$ is in state $0$, the states of $b_0$ and $b_1$ will remain $0$ and $1$ respectively, hence the iterations of $A_G$ from $x'$ will reproduce on $V_C$ the iterations of $C$ from input $x$ (there is a delay corresponding to the depth of the circuit $C$, where the computed values will advance layer-by-layer and be surrounded by states $0$ in all other vertices of $V_C$ thanks to the second item above).

    $(\Leftarrow)$ Assume that $C,x,i$ is a negative instance of $\IMCVP$.
    Then, the simulation of the iterations of $C$ will go on, and the output vertex number $i$ remains in state $0$, so $A_{G'}$ does not converges to the majority state $\majority{x'}=1$.

    $(\Rightarrow)$ Assume that $C,x,i$ is a positive instance of $\IMCVP$.
    Then, the simulation will eventually produce a state $1$ in the output vertex number $i$, immediately leading vertex $b_0$ to enter state $1$.
    Therefore, after $d$ iterations, with $d$ the depth of $C$, all the vertices in $V_C$ will also go to state $1$. 
    We conclude that all the vertices of $A_{G'}$ are in state $1$.
\end{proof}

From this result, we deduce that the naive algorithm consisting in checking
one by one whether each configuration correctly converges to the majority state,
is inherently slow to solve \BCTP{}.
More generally, all methods including a configuration test would face the \PSPACE{}-completeness of \OCCM{}. 
%At this point one might think that \BCTP{} is \PSC as well. 
Nevertheless, there may exist methods for solving \BCTP{} which are not based on any single configuration test. 
Indeed, the fact that all configurations must converge to a very specific fixed point may imply structural conditions on the MBAN itself. 
And these conditions might have a complexity less than \PSPACE{} to check.
This is what we are going to study in the next section.

%%%%%%%%%%%%%%%%%%%%%%%%%%%%%%%%
\section{Structural characterization}\label{sec:chara_MBAN_sol_DCTP}

Let $A_G$ be an MBAN.
A forbidden pattern is a structure $H$ such that,
if $G$ contains $H$, then $A_G$ cannot solve DCT. 
We prove that there exist three kinds of forbidden patterns:
leader, self-sufficient and self-sufficient $m$-cycle.
Furthermore, we prove that these patterns are sufficient to characterize 
the obstacles for an MBAN to solve DCT (Theorem~\ref{th:charac}).

\medskip
The first pattern captures instabilities in the majority state among iterations.

\begin{definition}\label{def:major}
     Let $G = (V, E)$ be a directed graph and $S$ be a subset of $V$. 
     We call \emph{major of $S$}, the subset $M(S)$ of $V$ such that $u \in M(S)$ if and only if $|N_G(u) \cap S| > \frac{|N_G(u)|}{2}$.
     That is, $u$ has a majority of its in-neighbors in $S$.
\end{definition}

The set-based notations $S\subseteq V$ and $M(S)\subseteq V$ are meant to represent
the set of components in state $1$ in some configuration.
Indeed, observe that if $x_u=1 \iff u\in S$, then $A(x)_u=1 \iff u\in M(S)$.
This property will be applied implicitly in the following.

\begin{definition}
    Let $G = (V, E)$ be a directed graph. 
    A subset of vertices $S$ is called \emph{leader} when
    $|S| < \frac{|V|}{2}$ and $|M(S)| > \frac{|V|}{2}$. 
\end{definition}

Intuitively, a leader subset of $V$ can change for at least one configuration $x$ the value of the majority between $x$ and $A(x)$. 

\begin{lemma}\label{lemma:leader_is_forbidden}
    Let $A$ be an MBAN defined on $G = (V,E)$.
    There is a configuration $x$ of $A$ such that $\majority{x} \neq \majority{A(x)}$
    if and only if there is a leader $S \subseteq V$.
\end{lemma}

\begin{proof}
    ($\Rightarrow$) Assume that $A$ admits a configuration $x$ such that $\majority{x} \neq \majority{A(x)}$.
    Without loss of generality, consider that $\majority{x} = 0$. 
    Let $S$ be the set of automata whose value is $1$ in $x$. 
    Thus, $|S| < \frac{|V|}{2}$. 
    Moreover, all automata whose value is $1$ in $A(x)$ have at least half of their predecessors in $S$.
    Consequently, each of them is in $M(S)$. 
    Since $\majority{A(x)} = 1$, it follows that $|M(S)| > \frac{|V|}{2}$ thus $S$ is a leader. 

	($\Leftarrow$) Assume that $G$ admits a leader $S$.
    We consider the configuration $x$ of $A$ such that $\val{x}{v}=1$ if and only if $v\in S$. 
    Then $\majority{x} = 0$. 
    However, since $S$ is a leader, we have $\val{A(x)}{u} = 1$ for all $u \in M(S)$,
    and since $|M(S)| > \frac{|V|}{2}$, it follows that $\majority{A(x)} = 1$. 
    %\qed
\end{proof}

The last two patterns capture cases where an individual automaton may not converge to the majority state, or where the system may not converge.

\begin{definition}\label{def:self-sufficient}
    Let $G=(V, E)$ be a directed graph. 
    A subset of vertices $S$ is called \emph{self-sufficient} when $S \subseteq M(S)$.
    We say that $S$ is a \emph{maximal self-sufficient} when $S = M(S)$. 
\end{definition}   

Intuitively, a self-sufficient subset of vertices may collectively hold and keep a common state,
whether it is the majority state of the whole network or not.

\begin{remark*}
  Self-sufficient subset of vertices appear under the name \emph{alliances} and \emph{defensive sets}
  in the literature, and it arises in various contexts.
  The notion is related to clustering coefficients, see for example~\cite{khh04,cmrs07} and the references therein.
\end{remark*}

We can see the notion of self-sufficiency and leadership as the opposite of each other. 
Indeed, a self-sufficient subset prevents the increase in majority value, and a leader subset causes the expansion in minority value.

\begin{lemma}\label{lemma:dense_SG_is_forbidden}
	Let $A$ be an MBAN defined on $G=(V,E)$ and $S$ be a subset of $V$ such that $|S| < \frac{|V|}{2}$.
    If $S$ is self-sufficient, then $A$ does not solve DCT.
\end{lemma}

\begin{proof}
	If $G$ admits a self-sufficient subset $S\subseteq V$ of size at most $\frac{|V|}{2}$,
	then for the configuration $x$ such that $\val{x}{v}=1$ if and only if $v\in S$,
	by Definition~\ref{def:self-sufficient} we have $A^t(x)_v = 1$ for all $v \in S$ and all $t \in \N$.
	However, the configuration $x$ contains a majority of $0$,
	hence $A$ does not solve DCT.
        %\qed
\end{proof}

Naively, we can think that the presence of self-sufficient is equivalent to
the presence of a non-trivial fixed point (that is, different from $0^n$ and $1^n$) in the dynamics. 
But, this is not always the case. 
Indeed, a self-sufficient can also be a leader. 
And more generally, if $S$ is self-sufficient, we can have $M^t(S) \subsetneq M^{t+1}(S)$,
for all $t$ until $M^t(S) = V$. 
However, the following holds.
%(proof in Appendix~\ref{a:proofs}).
%This is the reason why we define the notion of maximal self-sufficient. 

\begin{lemma}\label{lemma:chara_non_trivial_fixed_point}
    Let $A$ be an MBAN defined on $G=(V,E)$. 
    There exists a maximal self-sufficient $S \subsetneq V$ if and only if $A$ has a non-trivial fixed point.
\end{lemma}

\begin{proof}
    ($\Leftarrow$) Let us assume that the dynamic of $A$ has a non-trivial fixed point. 
    We call $x$ one of these fixed points, and we consider $V_0$ and $V_1$ the set of automaton whose value is respectively $0$ and $1$ in $x$. 
    Since the length of $x$ is odd, $\min(|V_0|,|V_1|) < \frac{|V|}{2}$. 
    Without loss of generality, we assume that $|V_1| < |V_0|$.
    Therefore, since $x$ is a fixed point, $M(V_1) = V_1$. 
    We conclude that $V_1$ is maximal self-sufficient.
    
    ($\Rightarrow$) Assume that there exists a maximal self-sufficient $S \subsetneq V$.
    First, note that since $S$ is maximal self-sufficient, the subset $V\setminus S$ is also maximal self-sufficient. 
    Indeed, by maximality of $S$, for any $u \notin S$ we have $|N_G(u) \cap S| \le \frac{|N_G(u)|}{2}$ and therefore $|N_G(u) \cap (V \setminus S)| > \frac{|N_G(u)|}{2}$.
    Thus, we can assume, without loss of generality, that $|S| < \frac{|V|}{2}$. 

    Consider $x$ the configuration of $A$ such that $\val{x}{v} = 1$ if and only if $v \in S$. 
    Since $M(S) = S$, it follows that all the automata which have value $1$ in $A(x)$ are in $S$. 
    We conclude that $x$ is a fixed point of the dynamics induced by $A$.
    %\qed
\end{proof}

We generalize self-sufficient subsets of vertices, when they have a similar but periodic behavior involving a sequence of subsets majoring one another cyclically.

\begin{definition}\label{def:self-sufficient-m-cycle}
	Let $G=(V, E)$ be a directed graph.
	An $m$-tuple of \textbf{distinct} subsets $(S_0,\dots,S_{m-1})\subseteq V^m$ is called a \emph{self-sufficient $m$-cycle}
	when $M(S_i) = S_{i+1 \mod m}$ and $S_i \neq V$ for all $0\leq i<m$.
\end{definition}

Note that, in the definition, the sets $S_0,\ldots, S_{m-1}$ are not necessarily disjoint. 

\begin{lemma}\label{lemma:dense_mDCG_is_forbidden}
    Let $A_G$ be an $MBAN$ on $G$, and $m>1$. 
    The dynamic $A_G$ has a limit cycle of length $m$ if and only if $G$ contains a self-sufficient $m$-cycle.
\end{lemma}

\begin{proof}
        $(\Rightarrow)$ Let us assume that the dynamics induced by $A$ contains a limit cycle $(u_1, \ldots, u_m)$ of length $m$, then $(S_0,\ldots, S_{m-1})$, with $S_i$ is the set of automaton with value $1$ in $u_i$, is a self-sufficient $m$-cycle.
        Indeed, the automaton with value $1$ at step $i$ have a majority of in-neighbors in state $1$ at step $i-1$, which exactly corresponds to the definition of major (Definition~\ref{def:major}).

        ($\Leftarrow$) Assume that $G$ contains a self-sufficient $m$-cycle $(S_0,\ldots,S_{m-1})$ for some $m > 1$.
        Consider the configuration $x$ such that $x_v=1$ if and only if $v\in S_0$.
        It follows from Definition~\ref{def:self-sufficient-m-cycle}
        that, for all $t \in \N$ :
        \[
            A^{t}(x)_{v}=\begin{cases}
                1 &\text{if } v \in S_{t \mod m},\\
                0 &\text{otherwise.}
            \end{cases}
        \]
        Thus, the dynamics induced by $A$ contains a cycle of length $m$.
        Indeed, $A^m(x) = x$ and for all integer $t \in \{1, \ldots, m-1\}$, all the automata whose value is $1$ in the configuration $A^t(x)$ are in $S_t$, implying that $A^t(x) \neq x$. 
        %\qed
\end{proof}

Observe that a self-sufficient $1$-cycle (for $m=1$) is simply a self-sufficient subset
(Definition~\ref{def:self-sufficient}). 
Moreover, if we consider $(S_0,\dots,S_{m-1})$ a self-sufficient $m$-cycle then $S_0 \cup \cdots \cup S_{m-1}$ is self-sufficient.  
Thus, if we assume that $S_0 \cup \cdots \cup S_{m-1}$ is maximal and does not contain all the nodes of $V$ ($S_0 \cup \cdots \cup S_{m-1} \neq V$), then $V \backslash (S_0 \cup \cdots \cup S_{m-1})$ is also self-sufficient and we go back to the previous pattern. 
For this reason, we can consider only the self-sufficient $m$-cycles which contain at least half of the nodes of the graph. 

\medskip
We now prove that these three patterns characterize exactly
whether or not an MBAN solves DCT (remark that they are not mutually exclusive).

\begin{theorem}\label{th:charac}
    Let $A$ be an MBAN  and $G = (V, E)$ its graph.
    Then, $A$ solves DCT if and only if $G$ does not contain any of the following patterns:
    \begin{itemize}[nosep]
        \item a leader subset $S$ (of size $|S| < \frac{|V|}{2}$),
        \item a maximal self-sufficient subset $S$ distinct from $V$,
        \item a self-sufficient-$m$-cycle $(S_0,\ldots,S_{m-1})$ (with $S_i \neq V$ for all $0\leq i<m$).
    \end{itemize} 
\end{theorem}

\begin{proof}
    One direction is given by Lemmas~\ref{lemma:leader_is_forbidden}, \ref{lemma:chara_non_trivial_fixed_point} and~\ref{lemma:dense_mDCG_is_forbidden}.
    For the other direction, suppose that $A$ does not solve DCT. 
    Thus, there exists a configuration $x$ of $A$ which is in one of the following three cases:
    \begin{enumerate}[nosep]
        \item $x$ has period at least $2$ in the dynamics induced by $A$,
        \item $x$ is a non-trivial fixed point of the dynamics induced by $A$,
        \item $\majority{x} \neq \majority{A(x)}$.
    \end{enumerate}
    And the three cases are also given by Lemmas~\ref{lemma:dense_mDCG_is_forbidden}, \ref{lemma:chara_non_trivial_fixed_point} and~\ref{lemma:leader_is_forbidden} respectively.
    %\qed
\end{proof}

%%%%%%%%%%%%%%%%%%%%%%%%%%%%%%%%
\section{Complexity of self-sufficient, leader and $m$-cycle}\label{sec:ss_and_lead_NPC}

We begin by proving that the existence of a self-sufficient pattern is an {\NPC} decision problem
(even when restricted to strongly connected digraphs).

\begin{theorem}\label{th:1_2_NPH}
  Given $G=(V,A)$ a digraph, deciding whether there exists a self-sufficient subset $S \subseteq V$ such that
  $|S| < \frac{|V|}{2}$ is {\NPC}. %, even when restricted to strongly connected digraphs.
\end{theorem}

\begin{remark*}
  The hardness of finding small self-sufficient subsets appears in the literature under the name
  \emph{alliances}, see for example~\cite{cmrs07,jhmr09}, and
  our construction is rooted in the same dual relationship to the hardness of finding large cliques.
  Nevertheless, our precise statement requires an original rework of known constructions,
  because it asks for the size of the subset to be strictly less than half of the vertices
  in the graph (in order to contradict the majority state).
\end{remark*}

\begin{proof}
  It belongs to {\NP}, because checking whether a potential subset $S\subseteq V$, with  $|S| < \frac{|V|}{2}$, is self-sufficient can straightforwardly be performed in polytime.

  \medskip
  For he \NP-hardness, we reduce from deciding whether an undirected graph $G=(V,E)$
  has a clique of size (at least) $\frac{|V|}{2}$, which is known to be {\NPH}~\cite{gj79}.
  Without loss of generality, assume $|V|$ is even
  and the degree of each vertex $v\in V$ is at most $|V|-3$.
  
  We construct a directed graph $G'=(V',A')$ with $V'=V\sqcup X\sqcup Y$ where $|X|=2|V|$ and $|Y|=2|V|+1$
  are sets of new vertices, so that $|V'|=5|V|+1$ is odd.
  The arcs of $A'$ are (see Figure~\ref{fig:1_2_NPH}):
  \begin{itemize}[nosep]
    \item all arcs toward $Y$ (including loops), \emph{i.e.} $\forall v\in Y,~|N_{G'}(v)|=5|V|+1$ is odd,
    \item all arcs from $X$ to $X$ (including loops) and $2|V|-1$ arcs from $Y$ to each vertex of $X$,
      and for a distinguished vertex $x\in X$ all arcs from $V$ to $x$,
      \emph{i.e.} $\forall v\in X\setminus\{x\},~|N_{G'}(v)|=4|V|-1$
      and $|N_{G'}(x)|=5|V|-1$ are odd,
    \item two arcs $(u,v),(v,u)$ for each $\{u,v\}\in E$ of $G$,
      for each $v\in V$ we add $|V|-3-|N_G(v)|$ arcs from $Y$ to $v$,
      plus $|V|$ additional arcs from $X$ to $v$,
      and $|V|$ additional arcs from $Y$ to $v$,
      \emph{i.e.} $\forall v\in V,~|N_{G'}(v)|=3|V|-3$ is odd.
  \end{itemize}
  When no precision is given, the in-neighbors can be chosen arbitrarily among the designated set of vertices.
  Observe that $X$ and $Y$ contain enough vertices for this construction,
  that it can be performed in polynomial time,
  and that the resulting digraph $G'$ is strongly connected,

  \begin{figure}[t]
    \centering
    \begin{tikzpicture}
      % sets and x
      \node[ellipse,minimum height=.8cm,minimum width=3cm,draw] (X) at (0,1.2) {$X$};
      \node[circle,minimum size=1.2cm,draw] (V) at (0,0) {$V$};
      \node[ellipse,minimum height=.8cm,minimum width=3cm,draw] (Y) at (0,-1.2) {$Y$};
      \node[circle,fill,inner sep=1pt,label=right:$x$] (x) at (-.6,1) {};
      % arcs
      \draw[-stealth,out=-40,in=50] (X) to node[right]{\scriptsize $|V|$} (V);
      \draw[-stealth,out=40,in=-50] (Y) to node[right]{\scriptsize $2|V|-3-|N_G(v)|$} (V);
      \draw[-stealth,out=220,in=140] (V) to node[left]{\scriptsize $|V|$} (Y);
      \draw[-stealth,out=140,in=-90] (V) to node[left]{\scriptsize $|V|$} (x);
      \draw[-stealth,out=178,in=183,pos=.6] (Y) to node[left]{\scriptsize $2|V|-1$} (X);
      \draw[-stealth,out=185,in=175,pos=.6] (X) to node[right]{\scriptsize $2|V|$} (Y);
      \draw[-stealth,loop right,looseness=6,out=4,in=-4] (X) to node[right]{\scriptsize $2|V|$} (X);
      \draw[-stealth,loop right,looseness=6,out=4,in=-4] (Y) to node[right]{\scriptsize $2|V|+1$} (Y);
      % in-degrees
      \draw[dotted] (X) -- (-3.5,1.2) node[left]{$4|V|-1$};
      \draw[dotted] (x) -- (-3.5,.8) node[left]{$5|V|-1$};
      \draw[dotted] (V) -- (-3.5,0) node[left]{$3|V|-3$};
      \draw[dotted] (Y) -- (-3.5,-1.2) node[left]{$5|V|+1$};
\end{tikzpicture}
    \caption{
      Construction of $G'$ in the proof of Theorem~\ref{th:1_2_NPH}.
      Arcs between the sets $X$, $V$ and $Y$ are labeled by the number of in-neighbors
      taken from the source set to all the vertices of the destination set.
      The total number of in-neighbors is given on the left, for each set and the distinguished vertex $x\in X$.
      %\vspace*{-1em} % HACK
    }
    \label{fig:1_2_NPH}
  \end{figure}
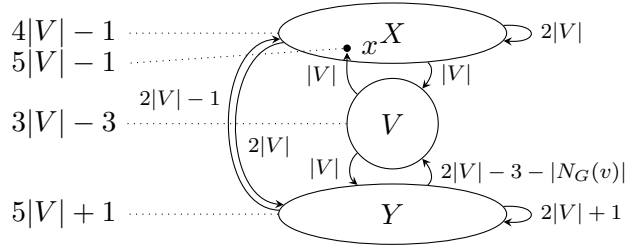

  We now show that $G$ has a clique of size at least $\frac{|V|}{2}$
  if and only if $G'$ has a self-sufficient subset of size at most $\frac{|V'|}{2}$.
  
  \medskip
  $(\Rightarrow)$
  If $G$ has a clique $S\subseteq V$ of size $|S|=\frac{|V|}{2}$
  (or more, but then it also has a clique of size exactly $\frac{|V|}{2}$),
  we consider $S'=S\cup X$ of size $|S'|=\frac{|V|}{2}+2|V|=\frac{|V'|-1}{2}$.
  For any $v\in X\setminus\{x\}$ we have
  $|N_{G'}(v)|=4|V|-1$ and $|N_{G'}(v) \cap S'|=2|V|$ (in-neighbors from $X$),
  and for the vertex $x$ we have $|N_{G'}(x)|=5|V|-1$ and $|N_{G'}(v) \cap S'|=2|V|+\frac{|V|}{2}$
  (in-neighbors from $X$ and from $S$, respectively).
  Moreover, for any $v\in S$ we have $|N_{G'}(v)|=3|V|-3$ and
  $|N_{G'}(v) \cap S'|=|V|+(\frac{|V|}{2}-1)=\frac{3|V|-2}{2}$
  (in-neighbors from $X$ and the half-clique, respectively).
  We deduce that $S'$ is a self-sufficient subset in $G'$.

  $(\Leftarrow)$
  If $G$ has no clique of size $\frac{|V|}{2}$ or more,
  let us consider a case analysis on any $S'\subseteq V'$ of size at most $\frac{|V'|-1}{2}$,
  to prove that $S'$ cannot be self-sufficient.
  \begin{itemize}[nosep]
    \item If $S'\cap Y\neq\emptyset$, then $S'$ is not self-sufficient because $|N_{G'}(v)|=|V'|$
      is too high for any $v\in Y$.
    \item If $S'\subseteq V$, then $S'$ is not self-sufficient because
      $|N_{G'}(v)|=3|V|-3$ whereas $|N_{G'}(v) \cap S'|\leq |V|-3$ for any $v\in V$.
  \end{itemize}
  The remaining case has $S'\subseteq V\cup X$ and $S'\cap X\neq\emptyset$.
  If $S'\cap X=\{x\}$ then $S'$ is not self-sufficient because the in-neighborhood of $x$ is too large
  ($5|V|-1$) compared to the size of $S'$ (at most $|V|+1$), hence
  there is another vertex $u\in S'\cap (X\setminus\{x\})$.
  Since $|N_{G'}(u)|=4|V|-1$, to be self-sufficient $S'$ requires $|N_{G'}(u) \cap S'|\geq 2|V|$,
  and this is possible only with $X\subseteq S'$ (there is no arc from $V$ to $u$).
  Given that $X\subseteq S'$, we deduce that $S'$ must additionally contain at least half of the vertices from $V$,
  to fulfill the condition of self-sufficient for the vertex $x$
  which has a total of $5|V|-1$ in-neighbors in $G'$.
  From the upper bound on the size of $S'$, we conclude that $S'$ must contain
  exactly half of the vertices from $V$, \emph{i.e.}~have size $\frac{|V'|-1}{2}$.
  Since $G$ has no clique of size $\frac{|V|}{2}$, it follows that for some $v\in V$ we have
  $|N_{G'}(v) \cap S'|\leq |V|+(\frac{|V|}{2}-2)$ (from $X$ and from $V\cap S'$, respectively)
  whereas $|N_{G'}(v)|=3|V|-3$, thus $S'$ is not self-sufficient.
  %\qed
\end{proof}

Remark that the proof also holds for maximal self-sufficient, hence by Lemma~\ref{lemma:chara_non_trivial_fixed_point} we obtain the following Corollary.

\begin{corollary}    
    Given a digraph $G$, 
    deciding whether the MBAN $A_G$ admits a non-trivial fixed point is {\NPC}. 
\end{corollary}

Deciding the existence of a leader forbidden pattern is also {\NPC}
(even when restricted to strongly connected digraphs). %, the proof is in Appendix~\ref{a:proofs}).

\begin{theorem}\label{th:leader_NPC}
    Given $G=(V,A)$ a digraph, deciding whether there exists a leader subset $S\subseteq V$ such that  $|S| < \frac{|V|}{2}$ is {\NPC}.
    %, even when restricted to strongly connected digraphs.
\end{theorem}

\begin{proof}%[Proof of Theorem~\ref{th:leader_NPC}]
    It belongs to $\NP$ because checking whether a potential subset $S \subset V$, with $|S|<\frac{|V|}{2}$, is leader can straightforwardly be performed in polynomial time. 

    \medskip
    We reduce from the problem of the existence of a vertex cover containing at most half the vertices, which is known to be {\NPC} (indeed, vertex covers and independent sets are closely related)~\cite{gj79}.  
    Let $G = (V,E)$ be an undirected connected graph with $V = \{0, \ldots, n-1\}$ where $n$ is a multiple of $4$. From the connectivity, we have that each vertex of $G$ has at least one neighbor.
    
    We construct the digraph $G'=(V',E')$ on vertex set $V' = V \sqcup Y \sqcup X \sqcup Z \sqcup Z'$ where:
    \begin{itemize}[nosep]
        \item $Y = Y_0 \sqcup \cdots \sqcup Y_{n-1}$ is such that $|Y_i| = |N_G(i)| - 1$ for all $i \in V$,
        \item $X = \{x_0, \ldots, x_{|Y|-1}\}$ \emph{i.e.}~$|X|=|Y|=2\,|E|-n \ge n$,
        \item $Z = \{\ell_0, r_0\} \sqcup \cdots \sqcup \{\ell_{3n/4-1}, r_{3n/4-1}\} $,
        \item $Z' = \{\ell_{3n/4},r_{3n/4}\} \sqcup \cdots \sqcup \{\ell_{n-1}, r_{n-1}\} \sqcup \{t\}$.
    \end{itemize}
    We have $|Z| = \frac{3n}{2}$, $|Z'| = \frac{n}{2} + 1$, and $|V'|=n+2\,(2\,|E|-n)+2\,n+1=n+4\,|E|+1$.
    
    The arcs of $E'$ are (See Figure~\ref{fig:leader_NPC}):
    \begin{itemize}[nosep]
        \item $E'$ contains $(i,j)$ and $(j,i)$ for each edge $\{i,j\}\in E$, 
        \item $N_{G'}(i)$ also contains all the nodes in $Y_i$ for each $i\in V$,
        \item $N_{G'}(x_i)=\{ x_{i-1}, \ell_i, r_i \}$ for each $i\in V$, where $x_{-1}$ is interpreted as $x_{n-1}$ (the vertices of $X$ form a cycle),
        \item $N_{G'}(u) = V'$ for all $u \in Y \sqcup Z$ \emph{i.e.}~they have a complete in-neighborhood,
        \item $N_{G'}(z') = X$ for all $z'\in Z'$.
    \end{itemize}
    Remark that $G'$ is strongly connected.

    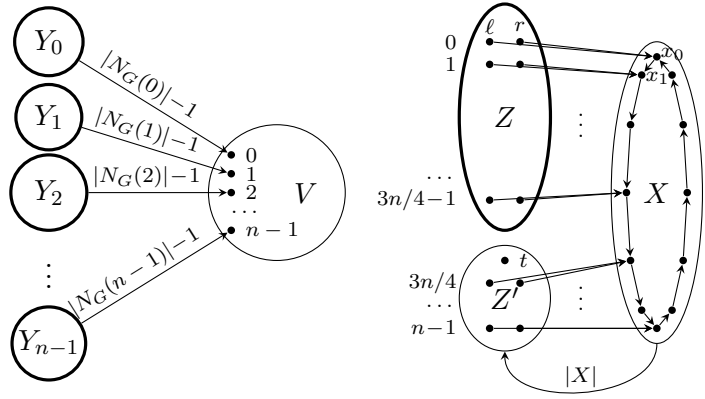
\begin{figure}[t]
        \centering
        \begin{tikzpicture}
  \def\dY{-3}
  \def\dZ{3}
  \def\dX{5}
  % sets Y_i
  \node[circle,minimum size=.9cm,draw,very thick] (Y0) at (\dY,2) {$Y_0$};
  \node[circle,minimum size=.6cm,draw,very thick] (Y1) at (\dY,1) {$Y_1$};
  \node[circle,minimum size=1cm,draw,very thick] (Y2) at (\dY,0) {$Y_2$};
  \node (Y) at (\dY,-1) {$\vdots$};
  \node[circle,minimum size=.6cm, inner sep=1pt,draw,very thick] (Yn-1) at (\dY,-2) {$Y_{n-1}$};
  % set V
  \node[circle,minimum size=1.8cm,draw] (V) at (0,0) {\qquad $V$};
  \foreach \y/\v in {1/0,.5/1,0/2,-1/n-1}
    \node[circle,fill,inner sep=1pt,label=right:\scriptsize$\v$] (\v) at (-.6,.5*\y) {};
  \node at (-.4,-.25) {\scriptsize \dots};
  % sets Z Z'
  \node[ellipse,minimum height=3cm,minimum width=1.2cm,draw,very thick] (Z) at (\dZ,1) {$Z$};
  \node[ellipse,minimum height=1.4cm,minimum width=1.2cm,draw] (ZZ) at (\dZ,-1.4) {$Z'$};
  \foreach \y/\v in {2/0,1.7/1,-.1/3nb4-1,-1.2/3nb4,-1.8/n-1}{
    \node[circle,fill,inner sep=1pt] (l\v) at (\dZ-.2,\y) {};
    \node[circle,fill,inner sep=1pt] (r\v) at (\dZ+.2,\y) {};
  }
  \node[circle,fill,inner sep=1pt,label=right:\scriptsize$t$] (t) at (\dZ,-.9) {};
  \node[left] at (\dZ-.5,2) {\scriptsize $0$};
  \node[left] at (\dZ-.5,1.7) {\scriptsize $1$};
  \node[left] at (\dZ-.5,.2) {\scriptsize \dots};
  \node[left] at (\dZ-.5,-.1) {\scriptsize $3n/4\!-\!1$};
  \node[left] at (\dZ-.5,-1.2) {\scriptsize $3n/4$};
  \node[left] at (\dZ-.5,-1.5) {\scriptsize \dots};
  \node[left] at (\dZ-.5,-1.8) {\scriptsize $n\!-\!1$};
  \node at (\dZ-.2,2.2) {\scriptsize $\ell$};
  \node at (\dZ+.2,2.2) {\scriptsize $r$};
  % set X
  \node[ellipse,minimum height=4cm,minimum width=1.2cm,draw] (X) at (\dX,0) {$X$};
  \foreach \r in {0,...,11}
    \node[circle,fill,inner sep=1pt] (x\r) at ($(\dX,0)+(30*\r+90:.4cm and 1.8cm)$) {};
  \node[xshift=6pt,yshift=1pt] at (x0) {\scriptsize $x_0$};
  \node[xshift=6pt,yshift=-1pt] at (x1) {\scriptsize $x_1$};
  % arcs from Y_i to V
  \foreach \v in {0,1,2,n-1}
    \draw[-stealth] (Y\v) to node[above,yshift=-2pt,sloped,pos=.42]{\scriptsize $|N_G(\v)|\!-\!1$} (\v);
  % arcs within X
  \foreach[count=\i from 0] \r in {1,2,3,4,5,6,7,8,9,10,11,0}
    \draw[-stealth] (x\i) -- (x\r);
  % arcs from Z Z' to X
  \foreach \z/\x in {0/0,1/1,3nb4-1/3,3nb4/4,n-1/6}{
    \draw[-stealth] (l\z) -- (x\x);
    \draw[-stealth] (r\z) -- (x\x);
  }
  \foreach \y in {1,-1.3}
    \node at (\dX-1,\y) {\scriptsize \vdots};
  % arcs from X to Z'
  \draw[-stealth,out=-90,in=-90] (X) to node[above,yshift=-2pt]{\scriptsize$|X|$} (ZZ);
\end{tikzpicture}
        \caption{
            Construction of $G'$ in the proof of Theorem~\ref{th:leader_NPC}.
            Arcs between the sets $Y_i$ and $V$ are labeled by the number of arcs.
            All $X$ is in neighbor of all $Z'$.
            The sets $Y_i$ and $Z$ have all the vertices in $V'$ as in-neighbors, hence there cannot belong to $M(S)$ for any leader subset $S$.
      %\vspace*{-1em} % HACK
        }
        \label{fig:leader_NPC}
  \end{figure}

    We now show that $G$ contains a vertex cover of size at most $\frac{|V|}{2}$  if and only if $G'$ contains a leader subset of size at most $\frac{|V'|}{2}$. 

    \medskip
    $(\Rightarrow)$ Let $C$ be a vertex cover of $G$ of size at most $\frac{|V|}{2}$. 
    Let $R$ be the set of nodes $r_i$ in $Z \cup Z'$ and $S = C \sqcup X \sqcup R$. 
    We have that $|S| \leq \frac{n}{2}+2\,|E|-n+n < \frac{|V'|}{2}$.  
    Since $S$ contains $X$ and $R$, we deduce that $X \subseteq M(S)$ and $Z' \subseteq M(S)$.
    Moreover, $C$ is a vertex cover of $G$ therefore $N_{G'}(i) \cap C = N_{G}(i)$ for all $i \in V$,
    and we have $|N_{G}(i)| = \frac{|N_{G'}(i)|+1}{2}$ for all $i \in V$, hence we deduce that $V \subseteq M(S)$. 
    Consequently, $|M(S)|\ge |X|+|Z'|+|V| = 2\,|E|-n+1+\frac{n}{2}+n > \frac{|V'|}{2}$ and $S$ is a leader subset of $V'$.

    $(\Leftarrow)$ Let $S$ be a leader subset of $V'$, hence of size $S<\frac{|V'|}{2}$
    and with major size $|M(S)|>\frac{|V'|}{2}$.
    We first prove that there exists a leader $S'$ avoiding $Y$ and using nodes from $R$ in $Z$ and $Z'$ (the set of nodes $r_i$ in $Z \sqcup Z'$). 

    \begin{claim}
        There exists a leader subset $S'\subseteq V'$ such that $S' \cap Y = \emptyset$ and $S' \cap (Z \sqcup Z') \subseteq R$. 
    \end{claim}
    
    \begin{claimproof}
        Firstly, we prove that it is possible to replace in $S$ each node which belongs to some $Y_i$ by some node belonging to $V$, or remove it, without changing $M(S)$.
        Assume that there exists $i \in V$ such that $S \cap Y_i \neq \emptyset$, and let $y$ be a node in this intersection.
        Two cases are possible.
        \begin{itemize}[nosep]
            \item If there exists $v \in N_{G}(i)$ such that $v \notin S$.
            Let $S'$ be the set $S$ where we replace $y$ by $v$. 
            Since no node in $Y$ nor $Z$ is in $M(S)$ (because their in-degrees are $|V'|$), we can consider only the out-neighborhood of $y$ included in $V$, namely the node $i$.
            We have $i\in N_{G'}(v)$, therefore $v$ compensates for the loss of $y$ and $M(S) \subseteq M(S')$.
            \item If $S \cap N_{G}(i) = N_{G}(i)$, 
            then $i \in M(S \setminus \{y\})$, 
            and for similar reasons as the previous item, $M(S) = M(S \setminus \{y\})$.
        \end{itemize}

        Secondly, we prove that $t$ is not in $S$, and that whenever both $\ell_i$ and $r_i$ are in $S$ then $\ell_i$ can be removed. Since in the case $\ell_i \in S$ and $r_i \notin S$, node $\ell_i$ can be replace by $r_i$, we conclude that $S' \cap (Z \sqcup Z')$ can be restricted to $R$.
        \begin{itemize}[nosep]
            \item If $t \in S$. 
            The out-neighborhood of $s$ is $Y \sqcup Z$ which cannot be in $M(S)$, thus we conclude that $M(S) = M(S \setminus \{s\})$.
            \item If there exists $i \in V$ such that $\ell_i, r_i \in S$ and $x_{i-1} \notin S$. 
            Let $S'$ be the set $S$ where we replace $\ell_i$ by $x_{i-1}$.
            For similar reasons as previously, we have that $M(S) = M(S')$.
            \item If there exists $i \in V$ such that $\ell_i, r_i \in S$ and $x_{i-1} \in S$.
            For similar reasons as previously, we have that $M(S) = M(S \setminus \{\ell_i\})$.
        \end{itemize}
        This concludes the proof of the claim.
    \end{claimproof}

    From the previous Claim, we assume without loss of generality that $S \cap Y = \emptyset$ and $S' \cap (Z \sqcup Z') \subseteq R$. 
    As already explained, $Y, Z \notin M(S)$, hence $M(S)$ contains only nodes of $X, Z'$ and $V$. 
    Moreover $|X| + |Z'| + |V| = \frac{|V'|+1}{2}$, hence it follows that $M(S) = X \sqcup Z' \sqcup V$ otherwise $S$ is not a leader. 
    
    For each $i\in\{n,\ldots,|X|-1\}$, in order to have $x_i \in M(S)$ we must have $x_{i-1}\in S$ because $N_{G'}(x_i)=\{x_{i-1}\}$.
    Furthermore we have $S \cap (Z \sqcup Z') \subseteq R$, meaning that no node $\ell_i$ is in $S$.
    As a consequence, for each $i\in\{0,\ldots,n-1\}$, in order to have $x_i \in M(S)$ we must have $x_{i-1},r_i \in S$. In total we must have $X,R \subseteq S$.

    To verify $|S|<\frac{|V'|}{2}$, at most $\frac{|V|}{2}$ node of $V$ are in $S$.
    Since no node of $Y$ is in $S$, in order to have $V\subseteq M(S)$, for each $i \in V$ we must have $N_{G}(i) \subseteq S$. 
    We conclude that $S \cap V$ is a vertex cover of $G$ whose size is at most $\frac{|V|}{2}$.
    %\qed
\end{proof}

We continue with studies on the complexity of deciding whether $G$ admits
a self-sufficient $m$-cycle pattern for some $m \geq 2$
(recall that $m=1$ corresponds to self-sufficient subset). %of Definition~\ref{def:self-sufficient}). 
Thanks to Lemma~\ref{lemma:dense_mDCG_is_forbidden}, this problem is equivalent to asking whether the dynamic of $A_G$ admits a limit cycle of length at least $2$. 

\begin{problem}
    {MBAN Limit Cycle problem}{$\SSmC$}{An MBAN $A_G$ given by its directed graph $G$}
    {Does the dynamics of $A_G$ contains a limit cycle of length $\geq 2$}
\end{problem}

To prove that $\SSmC$ is {\PSC}, we reduce from the \emph{Iterated Circuit Value Problem} ($\ICVP$), which is known to be {\PSC}~\cite{gmst16}.
An instance of $\ICVP$ is a circuit $C : \{0,1\}^n \to \{0,1\}^n$, a configuration $x\in\{0,1\}^n$ and a positive integer $i$ between $0$ and $n-1$, and the question is whether there is some integer $t>0$ such that $C^t(x)_i=1$, that is, the $t$-th iteration of the circuit $C$ starting from $x$ (feeding back the output to the input of the circuit at each iteration) outputs a bit $1$ on position $i$.
Without loss of generality, we assume that the circuit $C$ is \emph{layered}, that is to say the gates of depth $\ell$ take their inputs from gates at depth $\ell-1$ (depth $0$ being the inputs), and all the outputs are gates at the same depth $d$, which is the \emph{depth} of the circuit $C$ (this is achieved simply by adding $OR$ gates to copy values from one layer to the next). 
Our reduction towards an instance $G$ (a directed graph defining an MBAN) of $\SSmC$ is a construction in three steps.

First, we transform an instance of $\ICVP$, into a circuits $C_1$,
such that the dynamics induced by $C_1$ (as a function on state space $\{0,1\}^n$)
admits a limit cycle of length $2^n$ if and only if there exists no integer $t$ verifying that $C^{t}(x)_i = 1$,
otherwise $C_1$ has a unique attractor which is the fixed point $1^n$.
This is the purpose of Definition~\ref{def:global_atractor} and Proposition~\ref{prop:coICVP_reduce_EXPC}. 

Second, we give a method for transforming $C_1$ into a monotone Boolean circuit $C_2$
such that the dynamics induced by $C_2$ contains a limit cycle of length at least $2^n$
if and only if it is also the case for the dynamics of $C_1$,
otherwise both have only fixed points (two for the monotone Boolean circuit $C_2$, namely $0^n$ and $1^n$). 
This is done in Definition~\ref{def:circuit_to_monotone} and Proposition~\ref{prop:EXPCDO_reduce_MEXPC}.

Third, we explain how to transform a monotone Boolean circuit into an MBAN
(based on the gadgets in Figure~\ref{fig:gates}). 
Furthermore, we show that the dynamics induced by the MBAN,
resulting from the application of this transformation of $C_2$,
contains a limit cycle of length at least $2^n$ if and only if
it is also the case for the dynamics of $C_2$, otherwise both have only fixed points. 
We explain this last step in Definition~\ref{def:depth_one_to_MBAN} and Proposition~\ref{prop:CDOLC_reduce_SSmC}.

We also argue that all of the constructions can be performed in polynomial time. 
The proofs of Propositions~\ref{prop:coICVP_reduce_EXPC} and~\ref{prop:EXPCDO_reduce_MEXPC},
and of Lemmas~\ref{lemma:behavior_of_negated}, \ref{lemma:behavior_of_monotone} and~\ref{lemma:MBAN_TF}
are not particularly difficult from the definitions.
%they are presented in Appendix~\ref{a:proofs}.

Let us present the first step, constructing a circuit $L(C,x,i)$ with either a global fixed point attractor
if the $\ICVP$ instance is positive, or with a long cycle otherwise. %if the $\ICVP$ instance is negative.
This is achieved by a circuit incrementing a counter to obain a long cycle,
and sending the whole configuraiton to some fixed configuration whenever the application of circuit $C$
outputs a $1$ on its $i$-th bit.

\begin{definition}\label{def:global_atractor}
    %We define the transformation $L$, which constructs a Boolean circuit $C'$, from a Boolean circuit $C$, such that if an element of depth $j$ has for input a gate of depth $i$ with $i < j-1$ in $C$, then we add some $OR$ gates to copy the value in $i$ until the depth $j-1$, and also we add some $OR$ gates so that all the outputs of $C$ have depth $d$. 
    We define the transformation $L$ which, given $C,x,i$ an instance of $\ICVP$, constructs a Boolean circuits $L(C,x,i) : \{0,1\}^{2n} \to \{0,1\}^{2n}$ where its inputs and outputs are separated in two blocks each of $n$ bits, namely $y$ and $counter$, such that: 
    \begin{enumerate}[nosep]
        \item $1^{2n}$ is a fixed point,
        \item if $C(y)_i = 1$ then the output of $L(C,x,i)$ is $1^{2n}$,
        \item else, the $counter$ of $L(C,x,i)$ is incremented by $1$ modulo $2^n$,
        \item when the $counter$ equals $0^n$ then the output $y$ of $L(C,x,i)$ is $C(x)$, otherwise it is $C(y)$ (iterate $C$ on the input $y$). 
    \end{enumerate}
    We can construct $L(C,x,i)$ in polytime, as described in Figure~\ref{fig:fig_conv}. %Appendix~\ref{a:proofs} (Figure~\ref{fig:fig_conv}).
\end{definition}

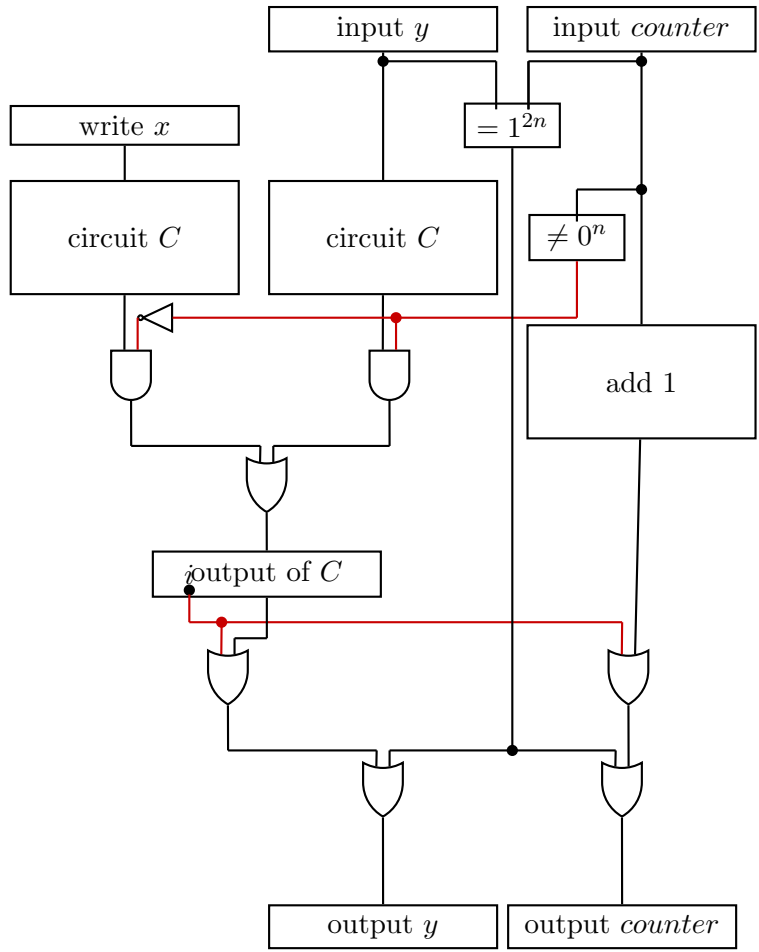
\begin{figure}
        \centering
        \begin{tikzpicture}[circuit logic US, thick, node distance=1cm and 2cm, scale=0.85]

    \tikzstyle{threadB} = [-, black]
    \tikzstyle{threadR} = [-, myRed]
    \tikzstyle{threadG} = [-, myGreen]
    \tikzstyle{inputs} = [draw, minimum width=3cm, minimum height=0.5cm]
    \tikzstyle{value} = [draw, minimum width=0.5cm, minimum height=0.5cm]
    \tikzstyle{module} = [draw, minimum width=1.25cm, minimum height=0.5cm]
    \tikzstyle{circuit} = [draw, minimum width=3cm, minimum height=1.5cm]
    \tikzstyle{AND} = [and gate, anchor=input 1, rotate=270]
    \tikzstyle{OR} = [or gate, anchor=input 1, rotate=270]
    \tikzstyle{XOR} = [xor gate, anchor=input 1, rotate=270]
    \tikzstyle{NOTV} = [not gate, anchor=input, rotate=270, scale=0.6]
    \tikzstyle{NOTH} = [not gate, anchor=output, rotate=180, scale=0.6]
    
    %lv 0
    \node[inputs] (A) at (0,8) {input $y$};
    \node[inputs] (counter) at (4,8) {input $counter$};
    \node[inputs] (W) at (-4,6.5) {write $x$};
    \node[module] (Una1) at (2,6.5) {$= 1^{2n}$};
    
    %lv 1
    \node[circuit] (C1) at (-4,4.75) {circuit $C$};
    \node[circuit] (C2) at (0,4.75) {circuit $C$};
    \node[module] (NUna0) at (3,4.75) {$\neq 0^n$};
    
    %lv 2
    \node[NOTH] (N1) at (-3.8,3.5) {};
    \node[AND] (A1) at (-3.8,3) {};
    \node[AND] (A2) at (0.2,3) {};
    \node[circuit] (Add) at (4,2.5) {add $1$};
    
    %lv 3
    \node[OR] (O1) at (-1.7,1.25) {};
    
    %lv 4
    \node[inputs] (U) at (-1.8, -0.5) {output of $C$};
    \node (I) at (-3, -0.5) {$i$};
    \node (IO) at (-3,-0.75) {};
    \filldraw (IO) circle (2pt);
    
    %lv 5
    \node[OR] (O2) at (-2.3,-1.75) {};
    \node[OR] (O5) at (3.9,-1.75) {};
    
    %lv 6
    \node[OR] (O3) at (0.1,-3.5) {};
    \node[OR] (O4) at (3.8,-3.5) {};

    %lv 7
    \node[inputs] (AO) at (0,-6) {output $y$};
    \node[inputs] (counterO) at (3.7,-6) {output $counter$};
    
    %thread inputs
    \draw[threadB] (W) to (C1);

    \coordinate (AP1) at (0,7.5);
    \coordinate (AP2) at (1.75,7.5);
    \coordinate (AP3) at (1.75,6.75);
    \filldraw   (AP1) circle (2pt);
    \draw[threadB] (A) to (AP1);
    \draw[threadB] (AP1) to (AP2);
    \draw[threadB] (AP1) to (C2);
    \draw[threadB] (AP2) to (AP3);

    \coordinate (counterP1) at (4,7.5);
    \coordinate (counterP2) at (2.25,7.5);
    \coordinate (counterP3) at (2.25,6.75);
    \coordinate (counterP4) at (4,5.5);
    \coordinate (counterP5) at (4,5.5);
    \coordinate (counterP5) at (3,5.5);
    \coordinate (counterP6) at (3,5);
    \filldraw   (counterP1) circle (2pt);
    \filldraw   (counterP4) circle (2pt);
    \draw[threadB] (counter) to (counterP1);
    \draw[threadB] (counterP1) to (counterP2);
    \draw[threadB] (counterP1) to (counterP4);
    \draw[threadB] (counterP2) to (counterP3);
    \draw[threadB] (counterP2) to (counterP3);
    \draw[threadB] (counterP4) to (counterP5);
    \draw[threadB] (counterP5) to (counterP6);
    \draw[threadB] (counterP4) to (Add);
    
    %thread circuits
    \draw[threadB] (C1) to (A1.input 2);
    \draw[threadB] (C2) to (A2.input 2);
    
    %thread B
    \coordinate (BP1) at (3,3.5);
    \coordinate (BP2) at (0.2,3.5);
    \coordinate (BP3) at (3.3,3.5);
    \filldraw[myRed] (BP2) circle (2pt);
    \draw[threadR] (NUna0) to (BP1);
    \draw[threadR] (BP1) to (BP2);
    \draw[threadR] (BP2) to (N1.input);
    \draw[threadR] (BP2) to (A2.input 1);
    \draw[threadR] (N1.output) to (A1.input 1);

    %thread AND lv 2
    \coordinate (A1P1) at (-3.9,1.5);
    \coordinate (A1P2) at (-1.9,1.5);
    \draw[threadB] (A1) to (A1P1);
    \draw[threadB] (A1P1) to (A1P2);
    \draw[threadB] (A1P2) to (O1.input 2);
    
    \coordinate (A2P1) at (0.1,1.5);
    \coordinate (A2P2) at (-1.7,1.5);
    \draw[threadB] (A2) to (A2P1);
    \draw[threadB] (A2P1) to (A2P2);
    \draw[threadB] (A2P2) to (O1.input 1);
    
    %thread OR lv 3
    \draw[threadB] (O1.output) to (U);
    
    %thread lv 4
    \coordinate (IP1) at (-3,-1.25);
    \coordinate (IP2) at (-2.5,-1.25);
    \coordinate (IP3) at (3.7,-1.25);
    \filldraw[myRed] (IP2) circle (2pt);
    \draw[threadR] (I) to (IP1);
    \draw[threadR] (IP1) to (IP2);
    \draw[threadR] (IP2) to (O2.input 2);
    \draw[threadR] (IP2) to (IP3);
    \draw[threadR] (IP3) to (O5.input 2);
    
    \coordinate (UP1) at (-1.8,-1.5);
    \coordinate (UP2) at (-2.3,-1.5);
    \draw[threadB] (U) to (UP1);
    \draw[threadB] (UP1) to (UP2);
    \draw[threadB] (UP2) to (O2.input 1);
    
    %thread OR lv 5
    \coordinate (O2P1) at (-2.4,-3.25);
    \coordinate (O2P2) at (-0.1,-3.25);
    \draw[threadB] (O2.output) to (O2P1);
    \draw[threadB] (O2P1) to (O2P2);
    \draw[threadB] (O2P2) to (O3.input 2);

    \draw[threadB] (Add) to (O5.input 1);
    \draw[threadB] (O5.output) to (O4.input 1);

    \coordinate (All1P1) at (2,-3.25);
    \coordinate (All1P2) at (0.1,-3.25);
    \coordinate (All1P3) at (3.6,-3.25);
    \filldraw (All1P1) circle (2pt);
    \draw[threadB] (Una1) to (All1P1);
    \draw[threadB] (All1P1) to (All1P2);
    \draw[threadB] (All1P1) to (All1P3);
    \draw[threadB] (All1P2) to (O3.input 1);
    \draw[threadB] (All1P3) to (O4.input 2);
    
    %thread AND lv6 
    \draw[threadB] (O3.output) to (AO);
    \draw[threadB] (O4.output) to (counterO);
    
\end{tikzpicture}
        \caption{
          Effective construction of $L(C,x,i)$ in Definition~\ref{def:global_atractor}.
          The test that the input is $1^{2n}$ is a circuit of depth $\log_2(2n)$
          with at most $2n$ $AND$ gates of fan-in two,
          and the test that $counter$ is not $0^n$ is analogous with $OR$ gates.
          Write $x$ is composed of $n$ NOT gates, $k$ $AND$ gates (for the $0$s in $x$)
          and $n-k$ $OR$ gates (for the $1$s in $x$).
          Add $1$ computes $counter + 1 \mod 2^n$ using $n$ $XOR$ gates and $n$ $AND$ gates.
          All the $OR$ and $AND$ gates drawn correspond to one gate per bit.
        }
        \label{fig:fig_conv}
\end{figure}

\begin{proposition}\label{prop:coICVP_reduce_EXPC}
     Let $C : \{0,1\}^n \to \{0,1\}^n$ be a Boolean circuit,
     $x\in\{0,1\}^n$ a configuration and $i \in \{0,\ldots, n-1\}$ a position.
     Then, $C,x,i$ is a negative instance of $\ICVP$ if and only if
     the iterations of $L(C,x,i)$ contains a limit cycle of length $2^n$. 
     Moreover, $C,x,i$ is a positive instance of $\ICVP$ if and only if
     the iterations of $L(C,x,i)$ has the global fixed point attractor $1^{2n}$. 
\end{proposition}

\begin{proof}
    %First remark that $L(C),x,i$ is a positive instance of $\ICVP$ if and only if it is %also the case for $C,x,i$. 
    %Let us consider $C,x,i$ be an instance of $\ICVP$. 
    For convenience we denote $C_1 := F(C,x,i)$.
    We prove one direction for each part of the statement, the other directions come from the conjunction of the two alternatives.

    \medskip
    Assume that $C,x,i$ is a negative instance of $\ICVP$,
    and consider the starting configuration of $C_1$ such that $y = C(x)$ and the value of $counter$ is the binary encoding of $1$.
    The iterations of $C_1$ will iterate $C$ on $y$ and increment $counter$ because it will never encounter a $1$ at position $i$ in $y$, until the $counter$ reaches the binary encoding of $0$, and then the output $y$ will be $C(x)$ and the output $counter$ will be the binary encoding of $1$, which is our starting configuration after $2^n$ steps.
    
    %We show by induction that $y = C^{counter + 1}(x)$.
    %At the end of the first iteration, since the value of $counter$ is $0$, then $y = C(x)$. 
    
    %Let $t \in \N$.
    %We assume that the properties are true for $t$, in other words, the output of $C_1$ restricted to $y$ is $C^{c+1}(x)$ and restricted to $counter$ is $c+1 \mod n$ with $c$ the value of $counter$ at the beginning of the iteration. Three cases are possible. 
    %\begin{itemize}[nosep]
    %    \item The value of $counter$ is $0$. 
    %    This case is similar to the base cases.
    %    \item  Secondly, the $t$-th output of $C_1$ restricted to $y$ is $C^{c+1}$ and restricted to $counter$ is in $\{1, \ldots, 2^n - 2\}$.
    %    Then, by construction, since $C^{c + 1}(x)_i \neq 1$ and $c + 1 \mod 2^n \neq 0$, the output of $C_1$ restricted to $y$ is $C^{c + 1}(x)$ and restricted to $counter$ is $c + 1$. 
    %    \item Thirdly, the $t$-th output of $C_1$ restricted to $y$ is $C^{c+1}$ and restricted to $counter = 2^n - 1$. 
    %Then, by construction, since $C^{c + 1}(x)_i \neq 1$ and $c + 1 \mod 2^n = 0$, the output of $C_1$ restricted to $y$ is $C^{c + 1}(x)$ and restricted to $counter$ is $0$. 
    %\end{itemize}
    %We conclude that the dynamic induced by $C_1$ admits a cycle with length $2^n$. 

    \medskip
    Assume that $C,x,i$ is a positive instance of $\ICVP$, with $C^t(x)_i=1$ for some integer $t$.
    Let $x'\in\{0,1\}^{2n}$ be any configuration of $C_1$, decomposed as $y'$ and $counter'$. 
    Two cases are possible. 
    First, if there exists $t'$ with $counter' + t' < 2^n$ such that $C^{t'}(y')_i = 1$, 
    then the value of $counter$ will not reset to $0$ during these iterations
    and $C_1^{t'}(x')$ is the fixed point $1^{2n}$.
    Second, if $C^{t'}(y')_i = 0$ for all $t'$ such that $counter' + t' < 2^n$, 
    then the output of $C_1^{k}(x')$ restricted to $y$ is $C(x)$ and restricted to $counter$ is $1$, with $k = 2^n - counter$. 
    Since $C^{t}(x)_i = 1$, then after $t-1$ more iterations we will have that $C_1^{k+t-1}(x')$ is the fixed point $1^{2n}$.
    %\qed
\end{proof}

The second step constructs a monotone circuit $M(C)$ with a behavior analogous to $C$.
This is achieved by employing the original circuit $C$ and its symmetric $N(C)$
according to the flip of $0$ and $1$ values.

\begin{definition}\label{def:circuit_negated}
    Let $C: \{0,1\}^n \to \{0,1\}^n$ be a Boolean circuit.
    We define the transformation $N$ which builds $N(C)$ from $C$ by replacing each $OR$ gate by an $AND$ gate and conversely.
    For any configuration $x\in\{0,1\}^n$, let $\overline{x}$ denote the configuration where each bit $x_i\in\{0,1\}$ of $x$ is flipped into $\overline{x}_i=\neg x_i$.
\end{definition}

\begin{lemma}\label{lemma:behavior_of_negated}
    Let $C: \{0,1\}^n \to \{0,1\}^n$ be a Boolean circuit.
    For any configuration $x\in\{0,1\}^n$, we have $N(C)(\overline{x})=\overline{C(x)}$.
\end{lemma}

\begin{proof}
    It simply follows by induction, from De Morgan's laws that for any two bits $b_1,b_2\in\{0,1\}$ we have
    $\neg b_1\wedge\neg b_2=\neg(b_1\vee b_2)$
    and
    $\neg b_1\vee\neg b_2=\neg(b_1\wedge b_2)$:
    at every gate the Boolean value computed in $N(C)$ is the negation of the Boolean value computed by the corresponding gate in $C$.
    Remark that $NOT$ gates also verify this (they are not modified in the construction of $N(C)$ presented in Definition~\ref{def:circuit_negated}).
    %\qed
\end{proof}

\begin{definition}\label{def:circuit_to_monotone}
    Let $C: \{0,1\}^n \to \{0,1\}^n$ be a Boolean circuit.
    We define the transformation $M$ which constructs from $C$ the monotone (without $NOT$ gate) Boolean circuit $M(C) : \{0,1\}^{2n} \to \{0,1\}^{2n}$ as follows.
    First, we delete any two consecutive $NOT$ gates in $C$.
    Second, we append a copy of $C$ and a copy of $N(C)$, therefore with $2n$ inputs and $2n$ outputs.
    Third, we remove all $NOT$ gates, and rewire them as follows:
    in $C$ (resp.~$N(C)$), if some input of gate $g$ was the output of a $NOT$ gate whose input was another gate $g'$, then we replace it with the gate corresponding to $g'$ in $N(C)$ (resp.~$C$).
    
    An example of the construction so far is given on Figure~\ref{fig:monotonize}.
    Note that this is a classical method to monotonize a Boolean circuit,
    using a double wire logic~\cite{bg92,gm96}.
    
    Furthermore, denoting $x^\ell\in\{0,1\}^n$ and $x^r\in\{0,1\}^n$ its two blocks of input bits, we set the following additional rules for the output of $M(C)$:
    \begin{enumerate}[nosep]
        \item\label{cons:test1} if $x^\ell$ is $1^n$ then the output is $1^{2n}$,
        \item\label{cons:test2} else, if $x^\ell_i=x^r_i=1$ for some $0\leq i<n$ then the output is $1^{2n}$,
        \item\label{cons:test3} else, if $x^\ell_i=x^r_i=0$ for some $0\leq i<n$ then the output is $0^{2n}$.
    \end{enumerate}
    An additional technical requirement for the final MBAN is that part $C$ (and $N(C)$) is layered,
    and that the part $C$ and the parts implementing Rules~\ref{cons:test1}--\ref{cons:test3}
    have coprime depths (the longer being $C$, so that the depth of $M(C)$ is the depth of $C$).
    A procedure to construct $M(C)$ is given in %Appendix~\ref{a:proofs}
    Lemma~\ref{lemma:monotone_construc_poly}.
\end{definition}

\begin{figure}[t]
    \centering
    \resizebox{.9\textwidth}{!}{\begin{tikzpicture}
    \tikzstyle{node} = [circle, draw, inner sep=1pt, minimum size=16pt]
    \tikzstyle{arc} = [-{>[length=1mm]}]
    % circuit
    \begin{scope}
      % input nodes
      \foreach \x in {0,1,2}
        \node[node] (i\x) at (\x,0) {$x_{\x}$};
      % not gates
      \node[node] (n0) at (0,-1) {$\neg$};
      \node[node] (n1) at (2,-1) {$\neg$};
      % output gates
      \node[node] (o0) at (0,-2) {$\wedge$};
      \node[node] (o1) at (1,-2) {$\vee$};
      \node[node] (o2) at (2,-2) {$\vee$};
      % arcs
      \draw[arc] (i0) to (n0);
      \draw[arc] (i0) to (o1);
      \draw[arc] (i1) to (o0);
      \draw[arc] (i1) to (o2);
      \draw[arc] (i2) to (o1);
      \draw[arc] (i2) to (n1);
      \draw[arc] (n0) to (o0);
      \draw[arc] (n1) to (o2);
      % labels
      \node[right=0pt of i2] {inputs};
      \node[right=0pt of o2] {outputs};
    \end{scope}
    % monotone circuit
    \begin{scope}[shift={(5,0)}]
      % input nodes
      \foreach \x in {0,1,2}
        \node[node] (il\x) at (\x,0) {$x^\ell_{\x}$};
      \foreach \x in {0,1,2}
        \node[node] (ir\x) at (3.5+\x,0) {$x^r_{\x}$};
      % output gates
      \node[node] (ol0) at (0,-2) {$\wedge$};
      \node[node] (ol1) at (1,-2) {$\vee$};
      \node[node] (ol2) at (2,-2) {$\vee$};
      \node[node] (or0) at (3.5+0,-2) {$\vee$};
      \node[node] (or1) at (3.5+1,-2) {$\wedge$};
      \node[node] (or2) at (3.5+2,-2) {$\wedge$};
      % arcs
      \draw[arc] (il0) to (ol1);
      \draw[arc] (il1) to (ol0);
      \draw[arc] (il1) to (ol2);
      \draw[arc] (il2) to (ol1);
      \draw[arc] (ir0) to (or1);
      \draw[arc] (ir1) to (or0);
      \draw[arc] (ir1) to (or2);
      \draw[arc] (ir2) to (or1);
      \draw[arc,dashed] (il0) to (or0);
      \draw[arc,dashed] (ir0) to (ol0);
      \draw[arc,dashed] (il2) to (or2);
      \draw[arc,dashed] (ir2) to (ol2);
      % labels
      \node[right=0pt of ir2] {inputs};
      \node[right=0pt of or2] {outputs};
      % brace C N(C)
      \draw[decorate,decoration={brace,amplitude=5pt}] (-.3,.5) -- node[above,yshift=5pt]{$C$} (2.3,.5);
      \draw[decorate,decoration={brace,amplitude=5pt}] (3.2,.5) -- node[above,yshift=5pt]{$N(C)$} (5.8,.5);
    \end{scope}
\end{tikzpicture}}
    \caption{
        First step of the construction of $M(C)$ in Definition~\ref{def:circuit_to_monotone}.
        From $C$ of size $n=3$ (left), one constructs a monotone circuit of size $2n=6$ simulating it (right). The final $M(C)$ has additional rules, but the following holds both for $M(C)$ and the pictured circuit: we have $C(010)=101$ and $M(C)(010\overline{010})=M(C)(010101)=101010=101\overline{101}$ (this illustrates Lemma~\ref{lemma:behavior_of_monotone}).
      %\vspace*{-1em} % HACK
    }
    \label{fig:monotonize}
\end{figure}
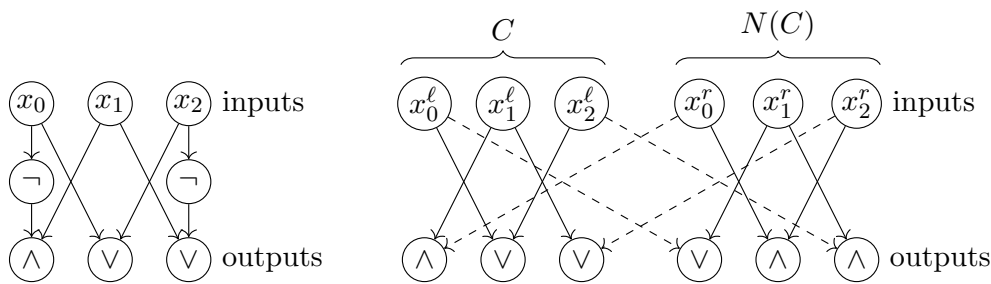

\begin{lemma} \label{lemma:monotone_construc_poly}
    Let $C: \{0,1\}^n \to \{0,1\}^n$ be a Boolean circuit. 
    The monotone Boolean circuit $M(C)$ can be constructed in polynomial time.
\end{lemma}

\begin{proof}
    First, a simple traversal allows to delete any two consecutive $NOT$ gates, construct $N(C)$, and rewire the removed $NOT$ gates.
    For Rules~\ref{cons:test1}--\ref{cons:test3}, we use additional gates organized as follows.

    To implement Rule~\ref{cons:test1} and verify whether $x^\ell$ is $1^n$, we use a binary tree of $AND$ gates of depth $\lceil \log_2 n \rceil$, and denote $r_1$ its root which has value $1$ if and only if $x^\ell=1^n$.

    We implement Rule~\ref{cons:test2} using $n$ $AND$ gates for the tests for each $0\leq i<n$, then a binary tree of $OR$ gates of depth $\lceil \log_2 n \rceil$ to propagate the finding of some $i$ such that $x^\ell_i=x^r_i=1$ with a value $1$ at the root gate, which we denote $r_2$.
    Rule~\ref{cons:test3} is implemented similarly, with $n$ $OR$ gates for the tests and a tree of $AND$ gates with a value $0$ at the root gate (in case $x^\ell_i=x^r_i=0$ for some $i$), which we denote $r_3$.
     
    The gates $r_1$, $r_2$ and $r_3$ are wired so the that the priority among Rules~\ref{cons:test1}--\ref{cons:test3} is respected: the $i$-th output gate $g$ of the circuit is replaced by the gate $g'$ which is computed as $g'=((g\wedge r_3)\vee r_2)\vee r_1$.
    Indeed, for all of these new output gates $g'$, if $r_1=1$ then $g'=1$, else if $r_2=1$ then $g'=1$, else if $r_3=1$ then $g'=0$. And if $r_1=r_2=0$ and $r_3$ (no rule applies) we have $g'=g$. An illustration is given on Figure~\ref{fig:monotone_construc}.
    
    \begin{figure}[t]
        \centering
        \resizebox{\textwidth}{!}{\begin{tikzpicture}
    \tikzstyle{node} = [circle, draw, inner sep=1pt, minimum size=14pt]
    \tikzstyle{arc} = [-{>[length=1mm]}]
    % monotone circuit C
    \begin{scope}
      % input nodes
      \foreach \x in {0,...,3}
        \node[node] (i\x) at (\x,0) {};
      % output nodes
      \foreach \x in {0,...,3}
        \node[node] (o\x) at (\x,-5) {};
      % boxes
      \draw (-.3,.3) rectangle node{$C$} (3.3,-5.3);
      % labels
      \node[left=6pt of i0] {inputs};
      \node[left=6pt of o0] {outputs};
      % depth
      \draw[stealth-stealth] (-.45,0) -- node[left]{$d$} (-.45,-5);
    \end{scope}
    % monotone circuit N(C)
    \begin{scope}[shift={(9,0)}]
      % input nodes
      \foreach \x in {4,...,7}
        \node[node] (i\x) at (\x,0) {};
      % output nodes
      \foreach \x in {4,...,7}
        \node[node] (o\x) at (\x,-5) {};
      % boxes
      \draw (4-.3,.3) rectangle node{$N(C)$} (7.3,-5.3);
    \end{scope}
    % ̈́Rule 1
    \begin{scope}[shift={(4,-3.5)},red]
      % tree of AND with root r1
      \foreach \x in {0,1,2,3}
        \node[node] (a\x) at (.75*\x,0) {$\wedge$};
      \foreach \x in {0,1}
        \node[node] (aa\x) at (1.5*\x+.75*.5,-.75) {$\wedge$};
      \node[node] (r1) at (.75*1.5,-1.5) {$\wedge$};
      \node[right=0pt of r1,black] {$r_1$};
      % arcs
      \foreach \g/\gg/\ggg in {%
        i0/i1/a0,i2/i3/a1,i4/i5/a2,i6/i7/a3,%
        a0/a1/aa0,a2/a3/aa1,aa0/aa1/r1%
      }{
        \draw[arc] (\g) -- (\ggg);
        \draw[arc] (\gg) -- (\ggg);
      }
    \end{scope}
    % ̈́Rule 2
    \begin{scope}[shift={(7,-3.5)},green]
      % n AND plus tree of OR with root r2
      \foreach \x in {0,1,2,3}
        \node[node] (e1\x) at (.75*\x,0) {$\wedge$};
      \foreach \x in {0,1}
        \node[node] (ee1\x) at (1.5*\x+.75*.5,-.75) {$\vee$};
      \node[node] (r2) at (.75*1.5,-1.5) {$\vee$};
      \node[right=0pt of r2,black] {$r_2$};
      % arcs
      \foreach \g/\gg/\ggg in {%
        i0/i4/e10,i1/i5/e11,i2/i6/e12,i3/i7/e13,%
        e10/e11/ee10,e12/e13/ee11,ee10/ee11/r2%
      }{
        \draw[arc] (\g) -- (\ggg);
        \draw[arc] (\gg) -- (\ggg);
      }
    \end{scope}
    % ̈́Rule 3
    \begin{scope}[shift={(10,-3.5)},blue]
      % n OR plus tree of AND with root r3
      \foreach \x in {0,1,2,3}
        \node[node] (e0\x) at (.75*\x,0) {$\vee$};
      \foreach \x in {0,1}
        \node[node] (ee0\x) at (1.5*\x+.75*.5,-.75) {$\wedge$};
      \node[node] (r3) at (.75*1.5,-1.5) {$\wedge$};
      \node[right=0pt of r3,black] {$r_3$};
      % arcs
      \foreach \g/\gg/\ggg in {%
        i0/i4/e00,i1/i5/e01,i2/i6/e02,i3/i7/e03,%
        e00/e01/ee00,e02/e03/ee01,ee00/ee01/r3%
      }{
        \draw[arc] (\g) -- (\ggg);
        \draw[arc] (\gg) -- (\ggg);
      }
    \end{scope}
    % new outputs
    \begin{scope}[shift={(0,-6)}]
      % nodes
      \foreach \y/\v in {0/\wedge,1/\vee,2/\vee}{
        \foreach \x in {0,...,3}
          \node[node] (o\x_\y) at (\x,-.75*\y) {$\v$};
        \foreach \x in {4,...,7}
          \node[node] (o\x_\y) at (9+\x,-.75*\y) {$\v$};
      }
      % arcs
      \foreach \x in {0,...,7}{
        \draw[arc] (o\x) -- (o\x_0);
        \draw[arc,blue] (r3) -- (o\x_0);
        \draw[arc] (o\x_0) -- (o\x_1);
        \draw[arc,green] (r2) -- (o\x_1);
        \draw[arc] (o\x_1) -- (o\x_2);
        \draw[arc,red] (r1) -- (o\x_2);
      }
      % label
      \node[left=6pt of o0_2] {outputs$'$};
    \end{scope}
\end{tikzpicture}}
        \caption{
            Construction of $M(C)$ described in Lemma~\ref{lemma:monotone_construc_poly}.
            The parts $C$ and $N(C)$ are split for clarity.
            The implementation of Rule~\ref{cons:test1} is in red,
            of Rule~\ref{cons:test2} in green,
            and of Rule~\ref{cons:test3} in blue.
            The depth $d$ of $C$ should be coprime with the depth of the tests until the roots $r_1$, $r_2$ and $r_3$, which equals $3$ in this picture.
            %\vspace*{-1em} % HACK
        }
        \label{fig:monotone_construc}
    \end{figure}
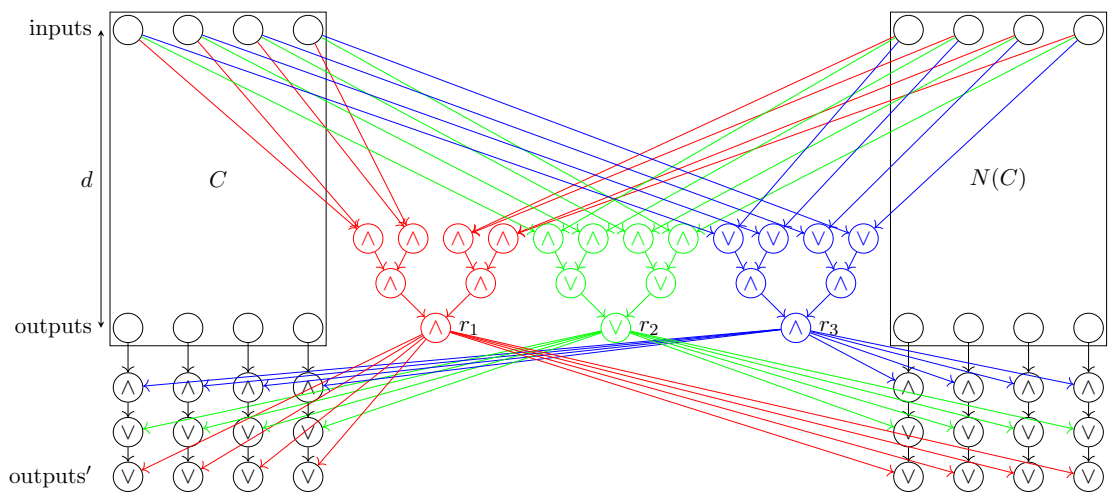

    If the initial circuit $C$ had depth $d$, it has now depth $d+3$. And the depth of the tests are $\lceil \log_2 n \rceil$ (binary tree for Rule~\ref{cons:test1}) or $\lceil \log_2 n \rceil+1$ (gates then binary tree for Rules~\ref{cons:test2} and~\ref{cons:test3}).
    Finally we add dummy $OR$ gates of arity one so that part $C$ is layered (split wires into layered bit propagations), and so that the depth of the tests and the depth of $C$ are coprime.
    %Finally, we add $p$ $OR$ gates, which simply copy the value outputted by the tests above, with $p$ chosen such that $\lceil \log_2 n \rceil + 3 + p$ is coprime with the depth of $C$.
    
    On overall, a polynomial number of gates are added, without any difficulty in the polynomial time computation of their placement.    
    %\qed
\end{proof}

To prove the second step of our construction (Proposition~\ref{prop:EXPCDO_reduce_MEXPC}), we need the following Corollary of Lemma~\ref{lemma:behavior_of_negated} and Definition~\ref{def:circuit_to_monotone}.

\begin{lemma}\label{lemma:behavior_of_monotone}
    Let $C: \{0,1\}^n \to \{0,1\}^n$ be a Boolean circuit.
    For any $x^\ell,x^r\in\{0,1\}^n$,
    if $x^r = \overline{x^\ell}$ then $M(C)(x^\ell x^r) = C(x^\ell)N(C)(\overline{x^\ell})= C(x^\ell)\overline{C(x^\ell)}$, 
    otherwise $M(C)(x^\ell x^r)$ is $0^{2n}$ or $1^{2n}$.
\end{lemma}

\begin{proof}
    We never encounter the configuration $1^{2n}$,
    therefore the result follows by induction over Lemma~\ref{lemma:behavior_of_negated}.
    %\qed
\end{proof}

\begin{proposition}\label{prop:EXPCDO_reduce_MEXPC}
    Let $C : \{0,1\}^n \to \{0,1\}^n$ be a Boolean circuit.
    For any $\ell>1$, there exists a limit cycle of length $\ell$ in the iterations of $C$
    if and only if there exists a limit cycle of length $\ell$ in the iterations of $M(C)$.  
    Moreover, the iterations of $C$ admits $1^n$ as a global attractor if and only if
    the only limit cycles in the iterations of $M(C)$ are the two fixed points $1^{2n}$ and $0^{2n}$.
\end{proposition}

\begin{proof}
    For convenience we denote $C_2 := M(C)$.
    We first prove one direction of the two equivalences, then the other directions.

    \medskip
    $(\Rightarrow)$
    For the first point, assume that there is a limit cycle of length $\ell>1$ in the dynamics induced by $C$.
    Let $x\in\{0,1\}^n$ be a configuration of this cycle.
    From Lemma~\ref{lemma:behavior_of_monotone}, we have $C_2^{t}(x\overline{x}) = C^{t}(x) \overline{C^{t}(x)}$ for all integer $t \ge 0$. 
    Consequently, the dynamics induced by $C_2$ admits a limit cycle of length $\ell$.  

    For the second point, assume that the dynamics induced by $C$ admits $1^n$ as a global attractor, and let us consider three cases.
    \begin{itemize}[nosep]
        \item If $x^\ell$ is $1^n$, then by construction the output of $C_2$ is $1^{2n}$.
        \item If $x^r = \overline{x^\ell}$, then we use the hypothesis that there exists $t$ such that $C^t(x^\ell)=1^n$,
        and an induction on Lemma~\ref{lemma:behavior_of_monotone} to deduce that $C_2^{t}(x\overline{x}) = C^{t}(x)\overline{C^{t}(x)} = 1^n 0^n$.
        Then the previous case applies.
        \item If $x^r \neq \overline{x^\ell}$, then by construction the output of $C_2(x^\ell x^r)$ is $0^{2n}$ or $1^{2n}$.
    \end{itemize}

    %\medskip
    $(\Leftarrow)$
    For the first point, assume that there is a limit cycle of $\ell>1$ in the dynamics induced by $C_2$.
    Let $y^\ell y^r$ be a configuration of this cycle.
    From Lemma~\ref{lemma:behavior_of_monotone} we deduce that $y^r = \overline{y^\ell}$, since otherwise the image of this configuration is a fixed point $1^{2n}$ or $0^{2n}$ (hence it is not in a cycle). 
    From Lemma~\ref{lemma:behavior_of_monotone} again, we have $C_2^{t}(y^\ell y^r) = C_2^{t}(y^\ell \overline{y^\ell}) = C^t(y) \overline{C^{t}(y)}$ for all integer $t\geq 0$.
    This implies that the first block of $n$ bits is always distinct throughout the iterations of $C_2$ for $\ell$ steps, hence the dynamics of $C$ admits a limit cycle of length $\ell$ as well.

    For the second point, assume that the iterations of $C_2$ admits as limit cycles only the two fixed points $1^{2n}$s and $0^{2n}$. 
    We can directly remark that all the configurations $x^\ell x^r$ that converge to $0^{2n}$ are such that $x^r \neq \overline{x^\ell}$.
    Then, for the configurations $x^\ell x^r$ such that $x^r = \overline{x^\ell}$, by hypothesis there exists $t \ge 0$ such that $C_2^t(x^\ell x^r)$ is $1^{2n}$.
    By a similar reasoning as the previous point, this implies that $C^t(x)$ is $1^n$, 
    and consequently $1^n$ is a global attractor of the iterations of $C$.
    %\qed
\end{proof}

We now construct the graph of an MBAN from a monotone circuit having only $OR$ and $AND$ gates,
using the simulation presented on Figure~\ref{fig:gates}.
Additionnal vertices $T$ and $F$ are meant to play the role of constants $1$ and $0$ respectively,
and we will start with the study of their behavior.

\begin{definition} \label{def:depth_one_to_MBAN}
    We define $B$ the transformation of a monotone Boolean circuit $C:\{0,1\}^n\to\{0,1\}^n$ into
    a graph $G =(V,E)$ with $V = C_0 \sqcup T \sqcup F$, where:
    \begin{itemize}[nosep]
        \item $C_0 =\{0, \ldots, s-1\}$ where $s$ is the number of gates in $C$ (excluding inputs),
        \item $T = \{s,s+1,s+2,s+3,s+4\}$ and $F = \{s+5,s+6,s+7,s+8,s+9\}$.
    \end{itemize}
    $E$ reproduces the adjacency among the gates of $C$ (without multiplicities), plus:
    \begin{itemize}[nosep]
        \item in-neighbor $s$ for any $OR$ gate with $2$ distinct input gates,
        \item in-neighbor $s+5$ for any $AND$ gate with $2$ distinct input gates,
        \item in-neighbor the $i$-th output gate of $C$, for any gate having as input the $i$-th input gate of $C$, with $0\leq i<n$.
    \end{itemize}
    Furthermore, the nodes $T$ and $F$ are two cliques (with loops). 
    Finally, $s+4$ and $s+9$ also have as in-neighbors the nodes $0$ and $1$. 
    Remark that $G$ is strongly connected, and all its nodes have in-degree $1$, $3$, $5$ or $7$. 
\end{definition}

Remark that $T$ and $F$ are self-sufficient subsets,
and we now prove that their states are fixed and uniform after two iterations.
This is important because the value of $s$ (from $T$) and $s+5$ (from $S$)
determine the simulation of Boolean gates by majority local functions.

\begin{lemma}\label{lemma:MBAN_TF}
    For any monotone circuit $C$, the states of $T$ and $F$ are uniform and fixed after two iterations of the MBAN $A_G$ on the digraph $G=B(C)$.
    More precisely, $\val{A_G^t(x)}{s}=\val{A_G^t(x)}{s+1}=\val{A_G^t(x)}{s+2}=\val{A_G^t(x)}{s+3}=\val{A_G^t(x)}{s+4}=\majority{x_T}$ and $\val{A_G^t(x)}{s+5}=\val{A_G^t(x)}{s+6}=\val{A_G^t(x)}{s+7}=\val{A_G^t(x)}{s+8}=\val{A_G^t(x)}{s+9}=\majority{x_F}$ for any $x\in\{0,1\}^n$ and any $t\geq 2$.
\end{lemma}

\begin{proof}
    We prove the lemma for $T$ (for $F$ it is analogous).
    The in-neighborhoods of $s,s+1,s+2,s+3$ are identical,
    and their value after one step of the MBAN is the majority state among $T$,
    \emph{i.e.}~$\majority{x_T}$. Since $\{s,s+1,s+2,s+3\}$ is self-sufficient, these values are fixed forever.
    At the second step, $s+4$ will also take this value,
    and $T$ is self-sufficient thus it keeps this uniform value forever.
    %\qed
\end{proof}

For the construction of Definition~\ref{def:depth_one_to_MBAN} to work properly, the vertices of $T$ (in particular $s$) are expected to be in state $1$, and the vertices of $F$ (in particular $s+5$) are expected to be in state $0$.
However, not all configurations of the MBAN verify this, hence we must also consider the other alternatives. %in our reasoning.
Intuitively, when the state of $s$ is $0$ then the $OR$ gates are simulated as $AND$ gates,
and when the state of $s+3$ is $1$ then the $AND$ gates are simulated as $OR$ gates.

To understand the dynamics of the MBAN on $M(C)$, it is easier when $C$ is layered.
However, in order to avoid creating artificial limit cycles,
we also need to break the fact that $C$ is layered,
which is done at the end of Definition~\ref{def:circuit_to_monotone}.
Consequently, for the next proposition we compose transformations $M$ and $B$.

\begin{proposition}\label{prop:CDOLC_reduce_SSmC}
    Let $C:\{0,1\}^n\to\{0,1\}^n$ be a Boolean circuit, $s$ be the number of gates in $M(C)$ and $d$ its depth,
    and $A_G$ be the MBAN on $G=B(M(C))$ with $s+9$ vertices.
    If there is a limit cycle of length $\ell>1$ in the dynamics induced by $C$ then there is a limit cycle of length $d\,\ell$ in the dynamics of $A_G$, and if there is a cycle of length $\ell'>1$ in the dynamics of $A_G$ then $\ell'=d\,\ell$ with $\ell>1$ and there is a cycle of length $\ell$ in the dynamics induced by $C$.
    Moreover, the dynamics induced by $C$ has the fixed point $1^n$ as a global attractor if and only if the only limits cycles in the dynamics of $A_G$ are fixed points.
\end{proposition}

\begin{proof}
    We split the study of the dynamics in $A_G$ for any $x\in\{0,1\}^{s+9}$ into four cases, depending on the majority state in $x$ of the vertices belonging to $T$ and $F$. Indeed, by Lemma~\ref{lemma:MBAN_TF} all the vertices of $T$ and $F$ converge after two steps of $A_G$, therefore we assume that it is the case in $x$
    (formally, for any starting $x'\in\{0,1\}$ we consider $x=A_G^2(x')$ in order to study the limit dynamics of $x'$).

    \medskip
    \textbf{Case $\majority{x_T}=\majority{x_F}=0$.}\quad
    In this case all the gates of $M(C)$ are simulated as $OR$ gates in $A_G$.
    Because $G$ is strongly connected and has cycles with coprime lengths (Definition~\ref{def:circuit_to_monotone}), it has only two fixed points on vertex set $C_0$
    ($T$ and $F$ are fixed): %(the only vertices whose state may evolve after two time steps):
    configuration $1^s$ if there exists $v\in C_0$ such that $x_v=1$, and configuration $0^s$ otherwise
    (see for example~\cite{gh00,g21}).

    \medskip
    \textbf{Case $\majority{x_T}=\majority{x_F}=1$.}\quad
    This case is analogous with the simulation of $AND$ gates only, and leads to the same conclusion (two fixed points).

    \medskip
    \textbf{Case $\majority{x_T}=1$ and $\majority{x_F}=0$.}\quad
    This is the expected behavior of $T$ (all its vertices are in state $1$) and $F$ (all its vertices are in state $0$) for the correct simulation of the $OR$ and $AND$ gates of $M(C)$ by $A_G$.
    Let us denote $r_1$, $r_2$ and $r_3$ the gates of circuit $M(C)$ which compute whether
    Rules~\ref{cons:test1}--\ref{cons:test3} of Definition~\ref{def:circuit_to_monotone} hold or not (respectively),
    as on Figure~\ref{fig:monotone_construc} from the polytime construction of Lemma~\ref{lemma:monotone_construc_poly}.

    \medskip
    If $C$ has a cycle of length $\ell>1$, then so does $M(C)$ by Proposition~\ref{prop:EXPCDO_reduce_MEXPC}. Let $\tilde{x}\in\{0,1\}^{2n}$ be a configuration of $M(C)$ on this cycle. We set the state of gate vertices $C_0$ in $x$ to the Boolean value they take on input $\tilde{x}$ in the circuit $M(C)$.
    
    It follows that the dynamics of $A_G$ simulates the computation of circuit $M(C)$ without applying any of the Rules~\ref{cons:test1}--\ref{cons:test3} of Definition~\ref{def:circuit_to_monotone}, meaning that $r_1=r_2=0$ and $r_3=1$ all along the computation, and we set the state of all the gates related to the test of Rules~\ref{cons:test1} and~\ref{cons:test2} to $0$, and the state of all the gates related to the test of Rule~\ref{cons:test3} to $1$.
    These vertices related to the tests keep their value all along the evolution of $A_G$
    and have no impact on the simulation of $M(C)$ by $A_G$.
    It takes $d$ steps to simulate one step of $M(C)$,
    so every $d$ steps the Boolean states of the vertices corresponding to the outputs of $M(C)$
    are a different configuration, meaning that the limit cycle in $A_G$ has length exactly $d\,\ell$.

    Now we prove that if $A_G$ has a limit cycle of length $\ell'>1$, then $\ell'$ is a multiple of $d$ and $M(C)$ has a limit cycle of length $\frac{\ell'}{d}>1$.

    \begin{claim}\label{claim:test_fail} % llncs gives no number ot claims...
        If there exists $t$ such that, in $A_G^t(x)$, the state of $r_1$ or of $r_2$ is $1$,
        then the dynamics of $A_G$ on $C_0$ converges to the fixed point $1^s$. 
        Else if the state of $r_3$ is $0$, then the dynamics of $A_G$ on $C_0$ converges to the fixed point $0^s$.
    \end{claim}

    \begin{claimproof}
      We analyse the step by step behavior of all the gates within circuit $M(C)$,
      as presented in the construction from Lemma~\ref{lemma:monotone_construc_poly} (see Figure~\ref{fig:monotone_construc}).
      If $r_1$ or $r_2$ are is state $1$, then it will eventually lead to a uniform
      configuration $1^{2n}$ on the vertices of $C_0$ corresponding to the output gates of $C$.
      This triggers the next test of Rule~\ref{cons:test1} to succeed, hence $r_1$ in state $1$,
      leading to another uniform configuration $1^{2n}$ on tge output gates of $C$, etc.
      Moreover, any $1^{2n}$ output configuration travels back into the circuit
      (as an MBAN, an input vertex of $C$ now has the corresponding ouput vertex of $C$ as in-neighbor),
      and by monotonicity it goes through each layer as a uniform configuration of $1$.
      Since the depth of the tests and the depth of the simulation of $C$ are coprime,
      this ultimately leads to the uniform value $1$ invading the entire set of vertices $C_0$,
      hence leading to a fixed point configuration $1^s$ on $C_0$.
    
      The reasoning is analogous when $r_1=r_2=r_3=0$, with uniform value $0$
      invading the entire set of vertices $C_0$.
      %
      %Indeed, this value leads to a uniform configuration on the output gates,
      %which triggers another test of Rules~\ref{cons:test1}--\ref{cons:test3} to fail,
      %which in turn triggers a uniform configuration on the output gates, etc.
      %A uniform input configuration to the part of $M(C)$ simulating circuit $C$ leads,
      %by monotonicity, to the same uniform output configuration,
      %and the applied rule also leads to the same uniform output configuration.
      %Since the depth of the test and the depth of the simulation of $C$ are coprime,
      %this ultimately leads to the uniform value invading the entire set of vertices $C_0$,
      %hence leading to a fixed point configuration $1^s$ or $0^s$ on $C_0$.    
    \end{claimproof}

    %The Claim is proven in Appendix~\ref{a:proofs}.
    From Claim~\ref{claim:test_fail},
    we deduce that in any configuration of a limit cycle of length $\ell'>1$ in $A_G$ the state of all the test vertices are fixed to $0$ for Rules~\ref{cons:test1} and~\ref{cons:test2}, and to $1$ for Rule~\ref{cons:test3}.
    As a consequence, the limit cycle occurs on the simulation of circuit $M(C)$, which is layered by Definition~\ref{def:circuit_to_monotone}, hence it corresponds to a limit cycle in the dynamics induced by $M(C)$, dividing by $d$ (the number of steps to simulate one step) the length $\ell'$ of the limit cycle in $A_G$. The resulting length $\frac{\ell'}{d}$ cannot be $1$, because it would actually be a fixed point in $A_G$.

    \medskip
    If $C$ has only the fixed point $1^n$ as a global attractor,
    then we consider two possibilities regarding the dynamics of $x$ in $A_G$.
    If some test of Rules~\ref{cons:test1}--\ref{cons:test3} fails, then by the Claim %Claim~\ref{claim:test_fail}
    the configuration $x$ converges to a fixed point.
    Otherwise, the dynamics of $A_G$ simulates the iterations of circuit $C$, and leads to the fixed point $1^n$ which triggers Rule~\ref{cons:test1}, hence it converges to a fixed point.

    \medskip
    If $A_G$ has only fixed points, then we consider only the configurations $x$ such that no test of Rules~\ref{cons:test1}--\ref{cons:test3} fails (otherwise the configuration is not legitimate to study the dynamics induced by $M(C)$ and $C$). All these configurations converge, meaning that the circuit simulation converges to a fixed point, thus $M(C)$ converge to a fixed point (which can only be $0^{2n}$ or $1^{2n}$), and by Proposition~\ref{prop:EXPCDO_reduce_MEXPC} the circuit $C$ admits $1^n$ as a global attractor.

    \medskip
    \textbf{Case $\majority{x_T}=0$ and $\majority{x_F}=1$.}\quad
    In this case all the $OR$ gates of $M(C)$ are simulated as $AND$ gates, and all the $AND$ gates of $M(C)$ are simulated as $OR$ gates, which produced a dynamics isomorphic to the previous case. 
    %\qed
\end{proof}

Combining Propositions~\ref{prop:coICVP_reduce_EXPC}, \ref{prop:EXPCDO_reduce_MEXPC} and~\ref{prop:CDOLC_reduce_SSmC}, we obtain that $\SSmC$ is {\PSC}, the reduction from {\ICVP} being the composition of transformations $G=B(M(L(C,x,i)))$, with $G$ of maximum in-degree $7$.
%From our constructions, we also have interesting corollaries. %(proof in Appendix~\ref{a:proofs}).

\begin{theorem}\label{th:SSmC_PSC}
    $\SSmC$ is {\PSC}.
\end{theorem}

From the intermediate steps of the construction, one can derive that many variants
of iterated circuit value problem are also {\PSC}.
These variations in the restrictions on the circuit and in the question
may be useful as a starting point for other hardness results,
therefore we state some meaningful ones explicitly as follows.

\begin{corollary}\label{coro:SSmC_PSC}
    The following problems are {\PSC}.
    \begin{itemize}[nosep]
        \item \emph{Circuit Global Fixed Point ({\CGlobalFP})}\\
        Given a Boolean circuit $C:\{0,1\}^n\to\{0,1\}^n$, does the dynamics induced by $C$ admits $1^n$ as a global attractor ?
        \item \emph{Monotone Circuit Two Fixed Points ({\MCtwoFP})}\\
        Given a monotone Boolean circuit $C:\{0,1\}^n\to\{0,1\}^n$, does the limit dynamics induced by $C$ admits only the two fixed points $0^n$ and $1^n$ ?
        \item \emph{Monotone Circuit Limit Cycle ({\MCLC})}\\
        Given a monotone Boolean circuit $C:\{0,1\}^n\to\{0,1\}^n$, does the dynamics induced by $C$ admits a limit cycle of length at least $2$ ?
        \item \emph{Monotone Circuit Exponential Limit Cycle ({\MCExpLC})}\\
        Given a monotone Boolean circuit $C:\{0,1\}^n\to\{0,1\}^n$, does the dynamics induced by $C$ admits a limit cycle of length at least $2^{n/2}$ ?
        \item \emph{Circuit Orbit $0^n$ to $1^n$ ({\Czerotoone})}\\
        Given a Boolean circuit $C:\{0,1\}^n\to\{0,1\}^n$, does there exist $t\in\N$ such that $C^t(0^n)=1^n$ ? That is, is $1^n$ in the orbit of $0^n$ ?
        %\item \emph{Circuit $k$ Connected Components} (C-$k$CC) for any $k\geq 1$\\
        %Given a Boolean circuit $C:\{0,1\}^n\to\{0,1\}^n$, does the dynamics induced by $C$ admits exactly $k$ connected components ?
        %\item \emph{Circuit at least $k$ Connected Components} (C-$\geq k$CC) for any $k\geq 2$\\
        %Given a Boolean circuit $C:\{0,1\}^n\to\{0,1\}^n$, does the dynamics induced by $C$ admits at least $k$ connected components ?
        \item \emph{Monotone Circuit $k$ Connected Components ({\MCkCC}) for any $k\geq 2$}\\
        Given a monotone Boolean circuit $C:\{0,1\}^n\to\{0,1\}^n$, does the dynamics induced by $C$ admits exactly $k$ connected components ?\\
        Also for $k=1$ without the monotonicity condition.
        \item \emph{Monotone Circuit $\geq k$ Connected Components ({\MCgekCC}) for any $k\geq 3$}\\
        Given a monotone Boolean circuit $C:\{0,1\}^n\to\{0,1\}^n$, does the dynamics induced by $C$ admits at least $k$ connected components ?\\
        Also for $k=2$ without the monotonicity condition.
    \end{itemize}
\end{corollary}

\begin{proof}
    The membership to $\PSPACE$ still follows a naive search algorithm as for {\OCCM}.
    The $\PSPACE$-hardness are obtained as follows.
    {\CGlobalFP} and {\Czerotoone} are immediate by Proposition~\ref{prop:coICVP_reduce_EXPC},
    which also gives the result {\MCkCC} for $k=1$ and {\MCgekCC} for $k=2$
    when the monotonicity condition is dropped.
    These three problems are trivial for monotonous circuits because there always is
    at least the two fixed points $0^n$ and $1^n$ in their limit dynamics.
    All other problems are in {\PSPACE} even when restricted to monotonous circuits.

    \medskip
    The transformation $M(L(C,x,i))$ proves,
    by Propositions~\ref{prop:coICVP_reduce_EXPC} and~\ref{prop:EXPCDO_reduce_MEXPC},
    the hardness of {\MCtwoFP}, {\MCLC}, {\MCExpLC}, {\MCkCC} for $k=2$ and {\MCgekCC} for $k=3$.
    One can artificially increase the number of connected components in the dynamics
    induced by a monotone circuit, by adding input/output bits.
    Given that the two original fixed points are $0^n$ and $1^n$,
    one can for example add $m$ bits which basically remain fixed
    (so far the number of fixed points is multiplied by $2^m$),
    except when the configuration has reaches one of the two original fixed points
    then the $m$ additional bits converge to a predefined number of fixed points
    (which may be different for the two fixed points,
    hence obtaining an arbitrary number of additionnal connected components).
    This settles the complexity of {\MCkCC} for any $k\geq 2$ and of {\MCgekCC} for any $k\geq 3$.
    %\qed
\end{proof}

\section{Conclusion and perspectives}

We have characterized when does an MBAN $A_G$ solves the DCT, by three forbidden patterns:
it does if and only if $G$ contains no leader (potentially changing the majority),
nor self-sufficient (potentially leading to a non-trivial fixed point),
nor self-sufficient $m$-cycle for any $m\geq 2$
(potentially leading to a limit cycle of length $\geq 2$).
Then, we have settled the complexity of recognizing each of these three patterns:
cheking whether $G$ admits a leader or a self-sufficient is {\NPC},
and whether $G$ admits a self-sufficient $m$-cycle with $m\geq2 $ is {\PSC}.
Corollary~\ref{coro:SSmC_PSC} gives larger consequences of interest.
Nevertheless, our results do not classify the complexity of \BCTP{}.
In particular, the construction of Theorem~\ref{th:SSmC_PSC} always has two self-sufficient subsets,
namely $T$ and $F$.

We conjecture that \BCTP{} is {\PSC}.
To design a polytime reduction, the current difficulty lies in constructing positive instances,
which is not easy since few families are known~\cite{ao20}.
A direction of research consists in the discovery of more graphs $G$ avoiding the three forbidden patterns,
\emph{i.e.}~such that $A_G$ solves the DCT.
Going beyond automata networks, one may further try to bring back some more uniformity to the construction,
for example $k$-regular graphs (for some small odd integer $k$)
so that the dynamical systems get closer to cellular automata.
More broadly, the discovery of other tasks that can be solved by some fixed local rule
in the spirit of the global tasks, depending on the structure of the graph, may be fruitful.

%%%%%%%%%%%%%%%%%%%%%%%%%%%%%%%%
\section*{Acknowledgment}

The authors received support from projects
ANR-24-CE48-7504 ALARICE and
HORIZON-MSCA-2022-SE-01 101131549 ACANCOS.

%%%%%%%%%%%%%%%%%%%%%%%%%%%%%%%%
\bibliographystyle{plain}
\bibliography{biblio}
	
%%%%%%%%%%%%%%%%%%%%%%%%%%%%%%%%
%\newpage
%\appendix

%%%%%%%%%%%%%%%%%%%%%%%%%%%%%%%%
%\section{Omitted proofs}\label{a:proofs}

\end{document}